\documentclass[lettersize,journal]{IEEEtran}
\usepackage{amsmath,amsfonts}
\usepackage{algorithmic}
\usepackage{algorithm}
\usepackage{array}
\usepackage[caption=false,font=normalsize,labelfont=sf,textfont=sf]{subfig}
\usepackage{textcomp}
\usepackage{stfloats}
\usepackage{url}
\usepackage{verbatim}
\usepackage{graphicx}
\usepackage{cite}
\usepackage{natbib}
\usepackage{doi}
\usepackage{enumitem}
\usepackage{multirow}
\usepackage{tabularx, booktabs}
\usepackage{amsthm}
\usepackage{hyperref}
\usepackage{fancyhdr}
\usepackage{amssymb}
\usepackage[table]{xcolor}
\usepackage{wrapfig}
\usepackage{fontawesome}
\usepackage{overpic}
\usepackage{subcaption}
\usepackage{tikz}
\usepackage{xurl}          
\hypersetup{breaklinks=true}
\makeatletter
\g@addto@macro{\UrlBreaks}{\do\@\do\_\do\.} 
\makeatother
\newtheorem{theorem}{Theorem}[section]
\newtheorem{proposition}{Proposition}

\newenvironment{remark}{\par\noindent\textbf{Remark.}\ }{\par}

\begin{document}


\title{Spectral Bottleneck in Sinusoidal Representation Networks: Noise is All You Need}

\author{%
  Hemanth Chandravamsi,
  Dhanush V. Shenoy,
  Itay Zinn,
  Ziv Chen, 
  Shimon Pisnoy, and
  Steven H. Frankel
  \thanks{{\raggedright
All the authors are with the Technion - Israel Institute of Technology, Haifa 3200003, Israel
(email: \url{hemanth@campus.technion.ac.il};
\url{dhanushv@campus.technion.ac.il};
\url{itay.z@campus.technion.ac.il};
\url{shimonpi@campus.technion.ac.il};
\url{ziv.chen@campus.technion.ac.il};
\url{frankel@me.technion.ac.il}).
\textit{Corresponding author: Hemanth Chandravamsi}.}
The code and data associated with this paper can be accessed through the project website at \url{https://cfdlabtechnion.github.io/siren_square/}.
The supplementary material corresponding to this manuscript is provided in the Appendix section.} 
}



\maketitle

\begin{abstract}
This work identifies and attempts to address a fundamental limitation of implicit neural representations with sinusoidal activation. The fitting error of SIRENs is highly sensitive to the target frequency content and to the choice of initialization. In extreme cases, this sensitivity leads to a spectral bottleneck that can result in a zero-valued output. This phenomenon is characterized by analyzing the evolution of activation spectra and the empirical neural tangent kernel (NTK) during the training process. An unfavorable distribution of energy across frequency modes was noted to give rise to this failure mode. Furthermore, the effect of Gaussian perturbations applied to the baseline uniformly initialized weights is examined, showing how these perturbations influence activation spectra and the NTK eigenbasis of SIREN. Overall, initialization emerges as a central factor governing the evolution of SIRENs, indicating the need for adaptive, target-aware strategies as the target length increases and fine-scale detail becomes essential. The proposed weight initialization scheme (WINNER) represents a simple ad hoc step in this direction and demonstrates that fitting accuracy can be significantly improved by modifying the spectral profile of network activations through a target-aware initialization. The approach achieves state-of-the-art performance on audio fitting tasks and yields notable improvements in image fitting tasks.
\end{abstract}


\section{Introduction}
\begin{figure*}[t!]
    \centering
    \includegraphics[width=\linewidth]{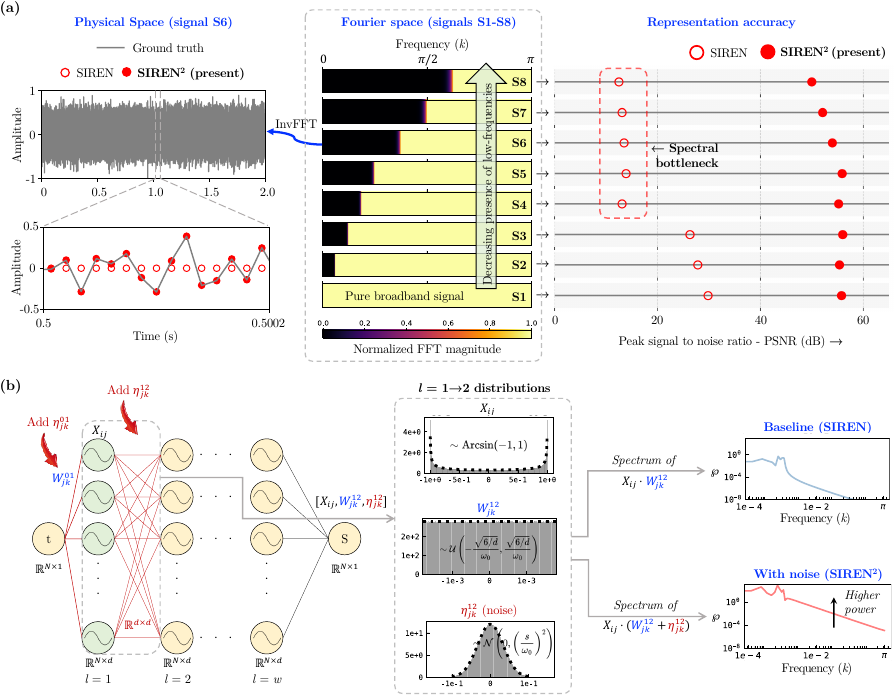}
    \caption{\textbf{Spectral bottleneck issue shown using SIREN and overview of the weight perturbation scheme.} (a) We attempt to fit eight discretely sampled broadband 1D signals (S1–S8) with decreasing low-frequency content. As shown in the right panel, the PSNR of SIREN \cite{sitzmann2020implicit} progressively decreases from S1 to S4, eventually encountering a spectral bottleneck after S4. SIREN fails to capture nearly all the frequencies of signals S4–S8, even though its frequency support should in principle allow it to represent a portion of the spectrum. In contrast, SIREN$^2$ initialized with WINNER maintains higher PSNR across all signals. (b) Schematic of a feedforward neural network with periodic activations, illustrating the statistical distributions of layer-1 outputs $X_{ij}$, weight matrix $W_{jk}$, and noise matrix $\eta_{jk}$. The effect of Gaussian noise ($\eta_{jk}$) on the spectrum of layer 2 pre-activations is shown: WINNER enhances the receptivity of high-frequencies.}
    \label{fig:abstract}
    \vspace{-15pt}
\end{figure*}
Implicit neural representations (INRs) of natural and synthetic signals have broad applications across various domains. They allow neural networks to represent coordinate-based discrete data such as images, videos, audio, 3D shapes, and scientific datasets as continuous functions. This allows seamless integration of multimedia, medical, and scientific data into machine learning pipelines for tasks such as denoising, classification, inpainting, and latent representation~\cite{mildenhall2021nerf,oechsle2019texture,niemeyer2020differentiable,szatkowski2022hypersound}. The key advantages of INRs over discretely sampled data are continuous input parameterization (no grid), fully differentiable networks with accessible gradients, and compressed data representation~\cite{sitzmann2020implicit, tancik2020fourier}. Using the readily available spatio-temporal gradients obtained through automatic differentiation, INRs can also be employed to solve forward and inverse problems governed by differential equations in a mesh-free setting~\cite{raissi2019physics, song2022versatile}.

Although multi-layer perceptrons (MLPs) are universal function approximators, training them to fit natural signals such as images, videos, or audio is often challenging. Analytical and empirical results are presented in previous studies \cite{rahaman2019spectral, xu2019frequency} showing that ReLU based deep neural networks are prone to `spectral bias' and are ill-conditioned for low-dimensional coordinate-based training tasks. To overcome spectral bias and the associated lazy training, coordinate inputs are often mapped to positional encodings~\cite{mildenhall2021nerf} or random Fourier features~\cite{rahimi2007random,tancik2020fourier}, which as shown in \cite{yuce2022structured} increase the `expressive power' of INRs. Alternatively, Sitzmann et al. \cite{sitzmann2020implicit} have proposed sinusoidal representation networks (SIRENs) using periodic activation function $\phi(x) = \sin(\omega_0 x)$,
\begin{equation} \label{eqn:siren}
\begin{aligned}
f^{\text{SIREN}}(\mathbf{x}; \theta) 
&= \mathbf{W}^{(L)} \mathbf{h}^{(L-1)} + \mathbf{b}^{(L)}, \\[4pt]
\mathbf{h}^{(l)} 
&= \phi^{\text{sin}}\!\left(\mathbf{W}^{(l)} \mathbf{h}^{(l-1)} + \mathbf{b}^{(l)}\right),
\quad \mathbf{h}^{(0)} = \mathbf{x}.
\end{aligned}
\end{equation}
where \( \mathbf{x} \in \mathbb{R}^d \) is the input, \( \theta = \{\mathbf{W}^{(l)}, \mathbf{b}^{(l)}\}_{l=1}^L \) are the network parameters, $L-1$ hidden layers, and \( \omega_0 \) is the activation periodicity specified as hyperparameter. They also propose a principled initialization scheme for the weights $\mathbf{W}$, to ensure network's pre- and post-activations at initialization remain narrowly bounded. Their approach faithfully captures both the discrete data and its gradients with high-fidelity, enabling applications in computer vision and the solution of scientific differential equations. Various other variants related to SIREN were also proposed, such as \cite{ziyin2020neural, liu2013fourier, silvescu1999fourier}.


While SIRENs with the uniform weight initialization scheme \cite{sitzmann2020implicit} can mitigate spectral bias and effectively fit 2D images, videos, and 3D geometries, they often perform poorly on low-dimensional signals such as audio when low-frequency components contribute minimally to the target. As pointed out in \cite{yuce2022structured}, \textit{the capacity of network to represent a target signal's frequency components does not guarantee efficient learning of that signal}. The reconstruction accuracy is rather closely tied to the spectral content of the signal and the inductive bias of network initialization; for instance, under fixed initialization, SIRENs struggle with signals dominated by very high or very low frequencies. Prior work \cite{mehta2021modulated, liu2024finer} shows that SIRENs exhibit overfitting as signal length and high-frequency content increase. An easy way to get around this problem is to increase the input dimensionality by mapping the input coordinates to a random Fourier feature space~\cite{rahimi2007random, tancik2020fourier, mildenhall2021nerf}. However, positional embeddings lead to a quadratic increase in the parameter count with respect to the embedding dimension, and consequently with the hidden-layer width. In this work, we examine the fitting challenges of SIREN on increasingly high-frequency targets and characterize the limitations of its traditional weight initialization \cite{sitzmann2020implicit}, particularly the \textit{spectral bottleneck} phenomenon, where SIREN fails to recover all frequency modes and collapses to a zero-valued function. We then propose a noisy weight initialization scheme WINNER that aims to mitigate the inductive bias at initialization. Figure~\ref{fig:abstract} illustrates the spectral bottleneck issue of the standard SIREN~\cite{sitzmann2020implicit} and the proposed weight initialization scheme which is being employed in SIREN$^2$. Representation accuracy of SIREN is compared with the proposed method on several synthetically generated audio signals, each containing 150{,}000 samples. Although, the spectral bottleneck phenomenon in Fig.~\ref{fig:abstract} is demonstrated through synthetic signals, such a phenomenon was noted to also observe in fitting natural signals, examples include \texttt{tap.wav} and \texttt{relay.wav} audio clips which can be accessed through the project's repository (link in first page). All experiments in Figure~\ref{fig:abstract} use a SIREN architecture with four hidden layers with 222 features in each layer.

\noindent The key contributions of this work are:
\begin{enumerate}[itemsep=0pt, parsep=0pt, partopsep=0pt, topsep=0pt]
    \item Show that SIRENs can suffer from a spectral bottleneck, and analyze their training dynamics in Fourier space to understand how this phenomenon unfolds during training.
    
    \item A new target-aware weight perturbation scheme WINNER, that adds noise into the baseline uniformly distributed weights of SIRENs to broaden their frequency support according to the target and thereby avoid spectral bottleneck. 
    
    \item The influence of the proposed noise addition scheme on the spectral distribution of pre-/post-activations and the eigenbasis of the empirical Neural Tangent Kernel (NTK) at initialization is analyzed. 
    
\end{enumerate}

\section{Related work}
\textbf{Activation functions for INRs.}  
Sinusoidal activation functions have been considered in neural networks since 1999 \cite{sopena1999neural}, but such models were often regarded as difficult to train \cite{parascandolo2016taming}. Sitzmann et al. \cite{sitzmann2020implicit} addressed this by introducing SIREN, which uses a tailored initialization scheme that enhances stable training of sinusoidal INRs. Later methods designed alternative periodic or adaptive activations \cite{jagtap2022deep} to extend frequency support, such as FINER \cite{liu2024finer}, FINER++ \cite{zhu2024finerplusplus}, and HOSC \cite{serrano2024hosc}, while others explored Gaussian or kernel-inspired nonlinearities \cite{ramasinghe2022beyond,chng2022gaussian}. DINER \cite{zhu2024disorder} further demonstrated that near-lossless signal representation can be achieved by augmenting a hash-table to a traditional INR backbone albeit with an additional memory overhead to store the hash data. Additional activation designs include Snake, a periodic trainable activation for learning periodic targets \cite{liu2020neural}; WIRE, which uses Gabor wavelet activations to capture frequency and spatial locality \cite{saragadam2023wire}; a sampling-theoretic case for using \(\mathrm{sinc}\) activations in INRs \cite{saratchandran2024sampling}; and STAF, a family of trainable sinusoidal activations \cite{morsali2025staf}.

\textbf{Weight/bias initialization.} Recent works like FINER~\cite{liu2024finer} and FINER++~\cite{zhu2024finerplusplus} explore bias initialization strategies to reduce eigenvalue decay in the empirical NTK, thus increasing effective frequency support. In geometry-focused INRs, SAL initializes network weights so that the signed distance function (SDF) starts as a 3D sphere~\cite{atzmon2020sal}; DiGS later introduced a multifrequency geometric initialization (MFGI) that extends this idea to sine activations and preserves high-frequency content~\cite{ben2022digs}; and in sinusoidal INRs, SIREN showed that without a tailored scheme, sine activations fail to reconstruct signals, motivating an initialization that stabilizes optimization~\cite{sitzmann2020implicit}. The works of Tang et al.~\cite{tang2025structured} and Varre et al.~\cite{varre2023spectral} (for linear networks) demonstrate how initialization can govern parameter optimization. Complementary strategies have been developed for sinusoidal and general INRs: VI$^3$NR derives an activation-agnostic, variance-preserving rule for INR regression~\cite{hewa2025vi3nr}; TUNER introduces frequency-sampling input initialization together with amplitude-bounded hidden-layer weights to control the model bandlimit during training~\cite{novello2025tuning}; and Yeom et al.~\cite{yeom2024fast} propose an initialization scaling for sinusoidal representation networks to speed up optimization and alleviate spectral bias. Prior designs are often specialized (geometry for SAL/DiGS and sine activations for SIREN) and typically do not target the joint distribution of pre- and post-activations across layers. In this work, we introduce a \textit{target aware weight initialization scheme that takes into account the spectral characteristics of the target} as well as preserves the layer distributions for stability and accuracy across INR tasks.

\textbf{Audio INRs.} A variety of models have been proposed to enhance the ability of neural networks to represent audio signals~\cite{choudhury2024nerva, kim2023regression, kim2023generalizable, li2024asmr}. Techniques such as hypernetworks, positional encoding, and auxiliary networks have also been explored to improve reconstruction fidelity and reduce reconstruction noise~\cite{marszalek2025hypernetwork, lanzendorfer2023siamese}. Audio-specific works include Siamese SIREN~\cite{lanzendorfer2023siamese} for compression, HyperSound~\cite{szatkowski2022hypersound} for INR generation via meta-learning, and INRAS~\cite{su2022inras} for spatial audio modeling.  Neural audio representations continue to find applications in classification, speech synthesis, sound event detection, encoding, and embedding~\cite{sharan2021benchmarking, tessarini2022audio, hershey2017cnn, donahue2018adversarial, shen2018natural, aironi2021graph}. 

In addition, analytical and empirical analysis of Basri et al. \cite{basri2020frequency, ronen2019convergence} indicate that for 1D signals, the convergence rate under gradient descent scales inversely with the square of the target signal's frequency (i.e., $1/k^2$), and this frequency-dependent slowdown grows exponentially with the input dimension~\cite{bietti2019inductive,cao2019towards}. This frequency-dependent limitation has motivated several INR-specific methods. For example, FR-INR reparameterizes MLP weights as linear combinations of fixed Fourier bases to mitigate low-frequency bias, while IGA estimates a gradient transform from a sampled empirical NTK and applies it during updates to accelerate learning of high-frequency components \cite{shi2024frinr,shi2024iga}.

\section{Understanding `Spectral Bottleneck'}

\subsection{Challenges of Fitting 1D Signals with SIREN} \label{sec:issues}
To maintain uniform feature scales and prevent exploding gradient issues, it is standard practice to normalize network inputs during training. However, for 1D data with high sampling rates, such as audio signals, input normalization introduces a mismatch between the frequency content of the signal and the frequency range effectively supported by the network. For instance, the input normalization \( \mathbf{x} \sim \mathcal{U}(-1, 1) \), scales the maximum frequency of a discretely sampled signal in proportion to its sampling rate. 

As for SIREN \cite{sitzmann2020implicit}, increasing the activation periodicity $\omega_0$ (Eqn \ref{eqn:siren}) may seem like a potential solution to mitigate the frequency bias induced by input scaling. However, since the pre-activation values of a SIREN predominantly lie within the range of $[-3,3]$ due to its weight initialization scheme, a high $\omega_0$ can render the activation function insensitive and alter the activation distribution. To address this, Sitzmann et al.~\cite{sitzmann2020implicit} scaled the inputs by a factor of 100, sampling them from a wider range $\mathbf{x} \sim \mathcal{U}(-100, 100)$, which effectively increases the power spectral density of the pre-activations by approximately the same factor (see Proposition~\ref{prop:uniform_psd_scaling}). While this strategy is effective for low-frequency-dominant signals, SIREN with input scaling can still fail on high-frequency-dominant signals. An example illustrating this failure is discussed in Fig.~\ref{fig:problemWSiren}.
\begin{figure}[t!]
    \centering
    \includegraphics[width=\linewidth]{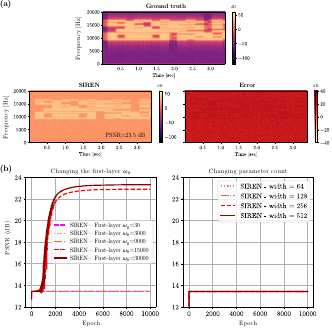}
    \caption{\textbf{An example case where the standard weight initialization scheme of SIREN fails to reconstruct an audio clip.} (a) Spectrograms of ground truth  (\texttt{tetris.wav}) (top), SIREN reconstruction (middle), and the error map (bottom) computed using first layer $\omega_0=30000$ and hidden layer width 128. (b,c) PSNR histories of SIREN for different input scalings and network sizes. A five-layer MLP was used with a learning rate scheduler reducing the rate by 2\% every 20 epochs from an initial value of $10^{-4}$.}
    \label{fig:problemWSiren}
\end{figure}

\begin{proposition} \label{prop:uniform_psd_scaling}
If $\mathbf{w}_a \sim \mathcal{U}(-a,a)^d$, then $z_a(\mathbf{x})=\mathbf{w}_a^\top \mathbf{x}$ can be written as $a\,z_1(\mathbf{x})$ with $\mathbf{w}_1 \sim \mathcal{U}(-1,1)^d$, hence its Fourier transform scales by $a$ and the power spectral density satisfies $S_a(k)=a^2 S_1(k)$.
\end{proposition}

\textbf{An example case when SIREN fails to fit a high-frequency signal:} We consider the reconstruction of a high-frequency audio signal \( f: \mathbb{R}^2 \to \mathbb{R}^d \), using the example \texttt{tetris.wav} available in the linked GitHub repository. Experiments conducted using SIREN in Fig.~\ref{fig:problemWSiren}(a,b) show that, although input scaling improves PSNR value, significant errors remain in the reconstructed signal, as shown in the spectrogram error. A maximum PSNR value of just $\sim$23.5dB was achieved using a scaling factor of $3\times 10^4$. We also experimented by increasing network size (Fig. \ref{fig:problemWSiren}b), changing scheduler parameters, and changing the frequency parameter $\omega_0$, none of which resolve this issue. This reveals that SIREN, with its default weight initialization, struggles to fit such high-frequency dominant signals due to the \textit{spectral bottleneck} phenomenon. Other 1D examples are shown in Fig.~\ref{fig:abstract}b, where SIREN exhibits a spectral bottleneck for signals S4–S8, all lacking low-frequency components.

\begin{remark}
    Repeating the experiment in Fig.~\ref{fig:problemWSiren} with alternative architectures, including WIRE \cite{saragadam2023wire}, FINER \cite{liu2024finer}, HOSC \cite{serrano2024hosc}, and Gauss \cite{ramasinghe2022beyond}, all yield the same outcome: convergence to a fixed PSNR of 13.4 with a zero-valued output. This pathology is therefore not just confined to SIRENs, but is shared across other related deep neural networks. We further observed that mapping the inputs to random Fourier features \cite{rahimi2007random, tancik2020fourier} or adopting a broader bias initialization, as proposed in \cite{zhu2024finerplusplus}, mitigates this failure mode. A formal characterization of this behavior is presented in Sec.~\ref{sec:ntksiren}.
\end{remark}

\subsection{Learning Dynamics in Fourier Space} \label{sec:ntksiren}

To explore the optimization trajectory that leads to the spectral bottleneck associated with SIRENs, as observed in Fig.~\ref{fig:problemWSiren}, we examine the network's learning dynamics in Fourier space through the lens of its empirical Neural Tangent Kernel (NTK) \cite{rahaman2019spectral, wang2022and, basri2020frequency, tancik2020fourier}. Before analyzing the training dynamics, we first define the NTK eigenbasis and discuss its interpretation. 

Jacot et al.~\cite{jacot2018neural} showed that in the infinite-width limit, fully connected networks trained with infinitesimal learning rates follow linear training dynamics governed by the NTK. The NTK is a Gram matrix defined as,
\begin{equation}
    \Theta(\mathbf{x}, \mathbf{x}') = \nabla_{\theta} \Phi(\mathbf{x}; \theta) \cdot \nabla_{\theta} \Phi(\mathbf{x}'; \theta),
\end{equation}
where $\Phi(\mathbf{x}; \theta)$ is the network output and $\theta$ its parameters. In the infinite-width regime, $\Theta$ remains constant during training, allowing for a closed-form linear model of training dynamics~\cite{lee2019wide}. The evolution of the output error $\mathcal{E}$ then satisfies,
\begin{equation}
    \frac{d \mathcal{E}}{dt} = -2 \Theta \mathcal{E}, \quad \Rightarrow \quad \mathcal{E}(t) = \mathcal{E}(\theta_0) e^{-2\Theta t}.
\end{equation} \label{eqn:ntk_governEqn}
\noindent This resembles a \textit{first-order rate equation} governing exponential decay, with $\Theta$ acting as the decay factor. To isolate the eigenmodes and their respective decay rates (eigenvalues), the NTK square matrix can be diagonalized as $\Theta = \mathbf{Q}^\top \Lambda \mathbf{Q}$. Substituting the diagonalized form into Eqn.~\ref{eqn:ntk_governEqn} and simplifying yields, $\mathcal{E}(t) = \mathbf{Q}^\top e^{-2\Lambda t} \mathbf{Q} \mathcal{E}(\theta_0)$. This equation implies that the components of reconstruction error $\mathcal{E}$ associated with larger eigenvalues decay more rapidly than those associated with smaller eigenvalues. Applying the Fourier transform to the error evolution equation yields:
\begin{equation} \label{eqn:error_fft}
    \hat{\mathcal{E}}(\theta_t) = \mathcal{E}(\theta_0) \hat{\mathbf{Q}}^{\top} e^{-2 \Lambda t} \hat{\mathbf{Q}}.
\end{equation}

where $\hat{\mathbf{Q}}$ denotes the Fourier-transformed eigenbasis of the NTK, and $\hat{\mathcal{E}}$ represents the error expressed in the frequency domain. This equation implies that, the Fourier components of error $\mathcal{E}$ are correlated to Fourier components of NTK eigenvectors $\hat{\mathbf{Q}}$. 

\begin{figure}[t!]
    \centering
    \includegraphics[width=\linewidth]{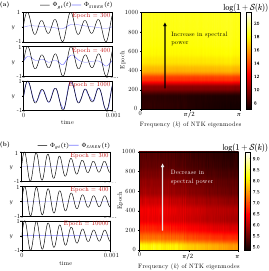}
    \caption{\textbf{Contrasting learning dynamics of SIREN for low and high frequency-dominant targets.} 
    (a) Output evolution when fitting the low-frequency signal of Eqn.~\ref{eqn:signals}. The right subplot shows the NTK spectral energy, $\log(1 + \mathcal{S}(k))$, across frequency eigenmodes $k$ during training, exhibiting a steady increase that indicates effective learning. 
    (b) Output and NTK spectral energy evolution for a high-frequency signal. The left subplot shows that SIREN fails to match the ground truth, while the right subplot reveals a suppression of spectral energy over training, indicating difficulty in representing high-frequency-dominant signals.}
    \label{fig:learnFreqs}
\end{figure}

Based on this premise, next we consider a toy problem where we attempt to supervise SIREN to fit two signals, one with predominantly low-frequency content, and another signal with predominantly high-frequency content. The signals are defined as follows:
\begin{equation}
    f(t) = 
\begin{cases}
\sum\limits_{i=1}^{3} A_k \sin(2\pi k_i^{(L)} t), & \text{(Low-frequency signal)} \\
\sum\limits_{i=1}^{3} A_k \sin(2\pi k_i^{(H)} t), & \text{(High-frequency signal)}
\end{cases}
\label{eqn:signals}
\end{equation}
\noindent The amplitudes are normalized as \( A = [\frac{1}{7}, \frac{2}{7}, \frac{4}{7}] \), and the frequencies are defined by \( k_{i=1,3}^{(L)} = [0.25\tilde{k}, \; 0.50\tilde{k}, \; 0.75\tilde{k}] \), \( k_{i=1,3}^{(H)} = [0.85\tilde{k}, \; 0.90\tilde{k}, \; 0.95\tilde{k}] \), with Nyquist frequency \( \tilde{k} = \frac{N}{2} \) and \( N = 2^{16} \) samples.

We use a five-layer SIREN with 128 hidden units, $\omega_0 = 30$, and inputs scaled by 10, giving a first-layer frequency of 300. The network is trained to fit the low- and high-frequency signals in Eqn.~\ref{eqn:signals}. During training, we track both the reconstructed signal and the cumulative eigenvalue-weighted NTK power spectrum:
\begin{equation}
    \mathcal{S}(k) = \sum_{i=1}^n \lambda_i \, \left|\widehat{\mathbf{v}}_i(k)\right|^2,
\end{equation}
where $\lambda_i$ are NTK eigenvalues and $\widehat{\mathbf{v}}_i(k)$ are the Fourier coefficients of the corresponding eigenvectors. $\mathcal{S}(k)$ measures the contribution of each frequency mode to the network's representation during training. 

\begin{figure}[t!]
    \includegraphics[width=\linewidth]{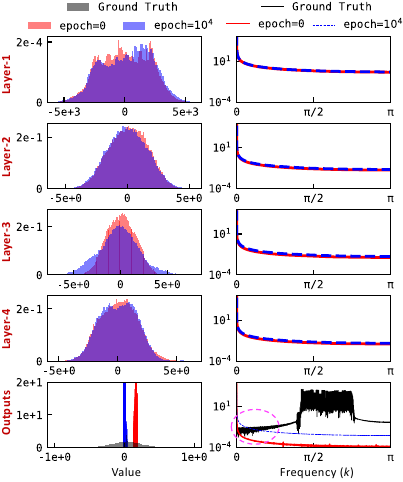}
    \caption{\textbf{Poor frequency support of SIREN for fitting high-frequency dominant targets.} Distributions (column-1) and cumulative power spectra (column-2) of hidden-layer pre-activation and network outputs for a four-layer SIREN at epoch $0$ and $10^{4}$ when fitting \texttt{tetris.wav}. Across all layers, the spectral content is concentrated far below that of the high-frequency target, indicating insufficient frequency support.}
    \label{fig:freq_support}
\end{figure}

As shown in Fig.~\ref{fig:learnFreqs}, SIREN accurately reconstructs the low-frequency target, with $\mathcal{S}(k)$ progressively increasing in the relevant modes. However, for the high-frequency target, the output remains near zero and $\mathcal{S}(k)$ is steadily suppressed, indicating an inability to fit high-frequency components. To address this phenomenon, we introduce a weight perturbation scheme in section~\ref{sec:scheme}.

\textbf{Frequency support.} We characterize the distributions and frequency response of pre-activations via the cumulative power spectral density (PSD) $\mathrm{PSD}(k) = \sum_{j=1}^{N_h} |\widehat{x}_{\mathrm{pre},j}(k)|^2$ of hidden-layer pre-activations. This experiment included the fitting of high-frequency data of \texttt{tetris.wav}. The signal contains 150{,}000 uniformly spaced samples in $[-1,1]$ and is evaluated using a four-layer SIREN. Fig.~\ref{fig:freq_support} shows that across all hidden layers, the value distributions remain centered and narrow, while the cumulative PSD is heavily biased toward low frequencies. The network outputs, along with the PSD of intermediate-layer pre-activations, exhibit a spectral profile that deviates significantly from the ground-truth spectrum, which is dominant of high-frequencies. Even after $10^{4}$ epochs, the mismatch persists, revealing a fundamental limitation: the effective frequency support of SIREN pre-activations falls far off from the target spectrum. This restriction acts as a learning bottleneck caused by the mismatch of spectral profiles of the network pre-activations and the target. We term such a scenario as `spectral bottleneck' and prevents the network from reconstructing signals like \texttt{tetris.wav}.

\begin{remark}
    As observed in Figs.~\ref{fig:learnFreqs} and \ref{fig:freq_support}, an extreme case of a spectral bottleneck is characterized by: (a) progressive attenuation of the NTK eigenmode energy with each training epoch, (b) failure of the network to represent even the low-frequency components (highlighted in Fig.~\ref{fig:freq_support}) that lie within the nominal frequency support of SIREN, and (c) convergence to a zero-valued output.
\end{remark}

\section{WINNER: Weight Initialization with Noise for NEural Representations} \label{sec:scheme}
A \textit{weight perturbation} scheme is proposed to address the spectral bottleneck of SIREN observed under its default weight initialization scheme~\cite{sitzmann2020implicit}. Traditional weight initializations, such as He and Glorot, aim to maintain the variance of pre- and post-activations within controlled bounds to avoid exploding or vanishing gradients, typically by scaling the weights inversely with the fan\_in or fan\_out. For example, the Glorot-style initialization used by Sitzmann et al.~\cite{sitzmann2020implicit} samples weights as $W_{jk} \sim \mathcal{U}\left(- \frac{1}{\omega_0} \sqrt{\frac{6}{\text{fan\_in}}}, \frac{1}{\omega_0} \sqrt{\frac{6}{\text{fan\_in}}} \right)$ to ensure unit standard deviation for all pre-activations ($\mathbf{W}\cdot\mathbf{x} + \mathbf{b}$) of the SIREN. 

In Sec.~\ref{sec:ntksiren}, we have emphasized that the root cause of the spectral bottleneck (Fig.~\ref{fig:learnFreqs}) is the mismatch of spectral energy between the target signal and network activations at initialization. We address this issue using Proposition~\ref{prop:uniform_psd_scaling} to alter the Fourier representation of the network's activations and outputs at initialization. To this end, we introduce WINNER, a weight perturbation scheme in which a Gaussian noise is added to the uniformly initialized weights that exist between the inputs and the second hidden layer. So, the weights immediately upstream of first and second hidden layers are perturbed as,
\begin{equation}
    W_{jk}^{(l)} \leftarrow W_{jk}^{(l)} + \eta_{jk}^{(l)},
    \label{eqn:perturbScheme}
\end{equation}
\noindent where the noise matrix \( \eta_{jk}^{l} \) is sampled from a normal distribution,
\begin{equation}
    \eta_{jk}^{(l)} \sim \mathcal{N}\left(0, \frac{s}{\omega_0}\right) \quad , \quad
    s =
    \begin{cases}
        s_0, & l=1, \\
        s_1, & l=2, \\
        0, & l=3,...,L. \\
    \end{cases}
\end{equation} \label{eqn:noiseScheme}
The Gaussian scale parameters \([s_0, s_1]\) control the width of the pre-activation distributions and their spectra. Based on the WINNER perturbation scheme, we introduce SIREN$^2$, a perturbed variant of SIREN, (the extra N in N$^2$ denotes \emph{Noise}).
\begin{equation} \label{eqn:siren_square}
\begin{gathered}
\text{SIREN}^2:\; f(\mathbf{x}; \theta)
= \mathbf{W}^{(L)} \mathbf{h}^{(L-1)} + \mathbf{b}^{(L)}, \\[4pt]
\mathbf{h}^{(l)} =
\begin{cases}
\mathbf{x}, & l = 0 \text{ (inputs)}, \\[3pt]
\phi^{\text{sin}}\!\big((\mathbf{W}^{(l)} + \eta_{jk}^{(l)}) \mathbf{h}^{(l-1)} + \mathbf{b}^{(l)}\big), & l = 1, 2, \\[3pt]
\phi^{\text{sin}}\!\big(\mathbf{W}^{(l)} \mathbf{h}^{(l-1)} + \mathbf{b}^{(l)}\big), & l = 3, \dots, L-1.
\end{cases}
\end{gathered}
\end{equation}
Although the weight perturbations are added only up to the second hidden layer, their effect propagates downstream all the way to the outputs. This is shown empirically in Sec.~\ref{sec:spectral} and in Supplementary Material Sec.~\ref{suppl:fulldists} for the full network. The goal of the proposed noise addition scheme is to enhance the functional sensitivity between the outputs and network parameters necessary to allow the parameter updates required for fitting high-frequency modes.


\begin{figure*}[t!]
    \centering
    \includegraphics[width=\linewidth]{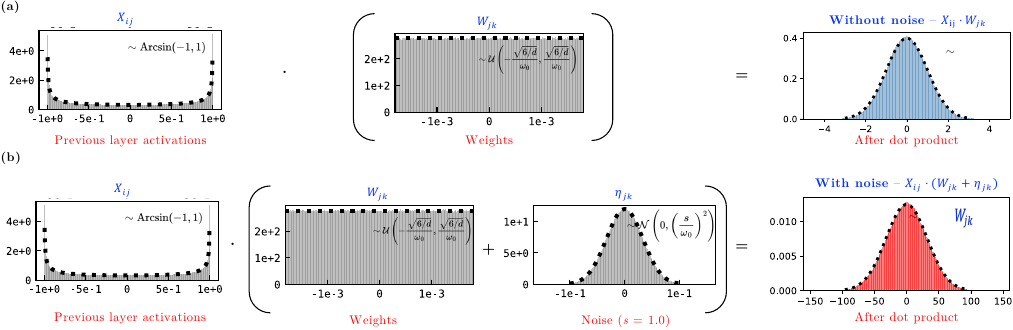}
    \caption{Distributions of the dot product between inputs $X_{ij} \sim \text{Arcsin}(-1,1)$ and weights initialized under two initialization schemes: (a) standard uniform weights $W_{jk} = \sim \mathcal{U}\left(- \omega_0^{-1}\sqrt{6/d}, \omega_0^{-1}\sqrt{6/d} \right)$, and (b) WINNER, in which uniform weights are perturbed with Gaussian noise, $W_{jk} + \eta_{jk}$. The noise addition increases the standard deviation of the dot product from $1$ to $\sqrt{1 + \frac{d s^2}{2}}$, where $d$ is the input dimension (fan\_in). This increase closely matches the analytically predicted value (Theorem~\ref{thr:theorem1}) shown by the dashed black line.}
    \label{fig:dists}
\end{figure*}

Fig.~\ref{fig:dists} illustrates the influence of the proposed noise perturbation on pre-activation distributions in a sinusoidal representation network. Assuming layer-1 activations ($X_{ij}$) follow an arcsine distribution on \((-1,1)\) (as established analytically and empirically in \cite{sitzmann2020implicit}), we compare the distributions of the dot product between $X_{ij}$ and weights connecting layer$1 \rightarrow2$ initialized with and without the noise $\eta_{jk}$. The results show that the overall structure remains consistent with the unperturbed network; Gaussian for both SIREN and SIREN$^2$. However, the added noise increases the standard deviation of the Gaussian pre-activations. This effect is analytically derived in Theorem~\ref{thr:theorem1} and empirically confirmed in Fig.~\ref{fig:dists}.

\begin{theorem} \label{thr:theorem1}
Let the following matrices be defined:
\begin{itemize}
    \item Input Matrix $X \in \mathbb{R}^{n \times d}$: Each entry $X_{ij}$ is independently sampled from an arcsine distribution $\mathcal{A}(-1, 1)$, which has a mean of $0$ and a variance of $1/2$.
    \item Weight Matrix $W \in \mathbb{R}^{d \times d}$: Each entry $W_{jk}$ is independently sampled from a uniform distribution $\mathcal{U}\left(-\frac{\sqrt{6/d}}{\omega_0}, \frac{\sqrt{6/d}}{\omega_0} \right)$ for a given $\omega_0 > 0$.
    \item Noise Matrix $\eta \in \mathbb{R}^{d \times d}$: Each entry $\eta_{jk}$ is independently sampled from a normal distribution $\mathcal{N}(0, (s/\omega_0)^2)$ for some scale parameter $s > 0$.
\end{itemize}
Consider the perturbed matrix $\omega_0 \cdot X (W + \eta) = Y'$, where $Y' \in \mathbb{R}^{n \times d}$. Then, for each entry \( Y'_{ik} \) of the matrix \( Y' \), its distribution is approximately Gaussian with zero mean and standard deviation \( \sqrt{1 + \frac{d s^2}{2}} \).
\end{theorem}

\begin{proof}
The random distribution $Y'$ can be decomposed as $Y' = Y_1 + Y_2$, with $Y_1 = \omega_0 (X W_{jk})$ and $Y_2 = \omega_0 (X \eta)$. Since $X$ and $W_{jk}$ are independent, the entries of $Y_1$ are sums of products of independent zero-mean random variables $X_{ij}$ and $W_{jk}$. Following \cite{sitzmann2020implicit} (see their Theorem 1.8) and by the central limit theorem (CLT), each $Y_{1,ik}$ is approximately Gaussian with mean and variance $\mathbb{E}[Y_{1,ik}] = 0$ and $\mathrm{Var}[Y_{1,ik}] = 1$ respectively. Similarly, for $Y_{2,ik} = \omega_0 \sum_{j=1}^d X_{ij} \eta_{jk}$, since $X_{ij}$ and $\eta_{jk}$ are independent and have zero-mean,
\begin{equation*}
\begin{aligned}
    \mathrm{Var}[X_{ij} \eta_{jk}] &= \mathbb{E}[(X_{ij} \eta_{jk})^2] - (\mathbb{E}[X_{ij} \eta_{jk}])^2 \\
    & = \mathbb{E}[X_{ij}^2] \cdot \mathbb{E}[\eta_{jk}^2] - 0 = \frac{1}{2} \cdot \left(\frac{s^2}{\omega_0^2}\right).
\end{aligned}
\end{equation*}
Now, for the sum,
\[
\sum_{j=1}^{d} \mathrm{Var}\left[X_{i j} \eta_{j k}\right]=\sum_{j=1}^{d} \frac{s^2}{2 \omega_0^2}=d \cdot \frac{s^2}{2 \omega_0^2}.
\]
Since $Y_{2,ik} = \omega_0 X_{i j} \eta_{j k}$, the variance scales by a factor of $\omega_0^2$, yielding,
\[
\mathrm{Var}[Y_{2,ik}] = \omega_0^2 \cdot \mathrm{Var}[X_{i j} \eta_{j k}] = \frac{d s^2}{2}.
\]
Finally, since $Y_1$ and $Y_2$ are independent, their variances add:
\[
\mathrm{Var}[Y'_{ik}] = \mathrm{Var}[Y_{1,ik}] + \mathrm{Var}[Y_{2,ik}] = 1 + \frac{d s^2}{2}.
\]
\[
\Rightarrow \mathbb{E}[Y'_{ik}] = 0, \quad \mathrm{Std}[Y'_{ik}] = \sqrt{1 + \frac{d s^2}{2}}.
\]

\end{proof}

\begin{proposition}
Given that the standard deviation of pre-activations in the layer-2 of a SIREN scale by a factor of $\sqrt{1 + \frac{d s^2}{2}}$ under the weight perturbation scheme in Eqn.~\ref{eqn:perturbScheme}, adding white noise $\eta_{jk}$ to the uniform weights $W_{jk} \sim \mathcal{U}\left(- \frac{1}{\omega_0} \sqrt{\frac{6}{\text{fan\_in}}}, \frac{1}{\omega_0} \sqrt{\frac{6}{\text{fan\_in}}} \right)$ is approximately equivalent to scaling the activation frequency $\omega_0$ by the same factor. This approximation holds under the assumption that the contribution of the bias vector to the pre-activation statistics is negligible, which is justified for large $d$ (fan\_in) values.
\end{proposition}

\subsection{Target-aware Specification of Noise Scales $s_0$ and $s_1$} \label{sec:specifyScales}
As shown in Figs.~\ref{fig:abstract}, \ref{fig:problemWSiren}, and \ref{fig:learnFreqs}, the performance of INRs is sensitive to the spectral content of the target signal. Therefore, an effective weight initialization is one that is target-aware and accounts for the spectral profile of the target. To quantify the weighted contribution of different frequencies to the target, a spectral centroid $\psi$ is defined as the normalized average frequency of the target's power spectrum, computed as
\begin{equation}
    \psi = 2 \times \frac{\sum_k k |\hat{y}(k)|}{\sum_k |\hat{y}(k)|},
\end{equation}
where \(\hat{y}(f_k)\) is the Fourier transform of the target evaluated at frequency bin \(k\). The factor $2$ normalizes \(\psi\) to the range \([0,1]\). Using \(\psi\), the noise scales \(s_0\) and \(s_1\) in Eqn.~\ref{eqn:noiseScheme} are empirically determined as  
\begin{equation}
    s_0 = s_0^{\text{max}}\!\left(1 - e^{-a\frac{\psi}{C}}\right), 
    \quad 
    s_1 = s_1^{\text{max}}\left(\frac{\psi}{C}\right),
    \label{eqn:s0s1scheme}
\end{equation}
where \(C\) denotes the number of channels (e.g., $C=3$ for RGB images). The hyperparameters \([s_0^{\text{max}}, a, s_1^{\text{max}}]\) are set as \(\left[\frac{3500}{70^{N_d-1}}, 5, \frac{4}{10^{N_d-1}}\right]\), with $N_d$ representing the input dimension of the network. Eqn.~\ref{eqn:s0s1scheme} is \textit{valid for all modalities} including audio, images, video, and 3D data.

The functional form of \(s_0\) in Eqn.~\ref{eqn:s0s1scheme} is empirically determined through multiple experiments. An initial linear formulation was tested, but experiments relating the spectral centroid to optimal \(s_0\) values revealed a trend better captured by an exponential curve rather than a linear one (see Supplementary Fig.~S1). The linear dependence was retained for \(s_1\). Both expressions are therefore ad-hoc and derived from multiple empirical evaluations; further work is required to develop a principled, target-aware formulation. Additional plots illustrating the sensitivity of noise scales to fitting accuracy for various audio clips are included in the supplementary material. As also detailed in \cite{sitzmann2020implicit}, the larger input scaling factor or the larger noise scales used for audio signals arise from their high sampling rates and higher frequency content compared to visual data.


Fig.~\ref{fig:sense} reports the effect of varying $s_0$ and $s_1$ on audio and image reconstruction. We find that SIREN$^2$ maintains strong performance over a broad range of values, showing that the method is not overly sensitive to precise tuning. The cross-marked settings from Eqn.~\ref{eqn:s0s1scheme} reliably fall in regions of high PSNR, providing a simple and robust rule for setting the perturbation scales without expensive hyperparameter search.

\begin{figure}[t!]
    \includegraphics[width=\linewidth]{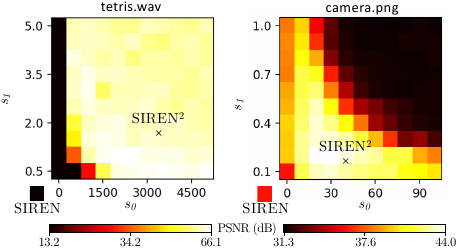}
    \caption{\textbf{Sensitivity to noise scales.} PSNR for audio (left) and image (right) reconstruction as a function of the perturbation scales $s_0$ and $s_1$. Performance is stable across a wide range, with the cross-mark indicating the target-aware scales from Eqn.~\ref{eqn:s0s1scheme} that consistently yield near-optimal results.}
    \label{fig:sense}
\end{figure}

\textit{SIREN$^2$ vs. FINER++:}
The current work compares against FINER/FINER++ \cite{liu2024finer, zhu2024finerplusplus}, which also aim to improve high-frequency representation but do so by modifying the bias initialization and the activation function. However, altering biases does not mathematically guarantee Gaussian-distributed pre-activations, and therefore cannot inherit the properties of Gaussian initialization that stabilize optimization and control signal propagation \cite{glorot2010understanding}. Another key difference is that FINER++ does not modify the weights connecting the first and second hidden layers, which limits its ability to adjust the initial spectral profile of pre-activations downstream of the first hidden-layer.

\section{Spectral Properties of SIREN$^\mathbf{2}$} \label{sec:spectral}
\subsection{Neural Tangent Kernel at Initialization} \label{sec:ntk_spectrum}
The spectral characteristics of the NTK at initialization are examined with an aim to understand the frequency support of SIREN$^2$ in comparison to SIREN. As shown in Fig.~\ref{fig:ntkspecs}, SIREN exhibits rapid eigenvalue decay and allocates higher cumulative spectral energy $\mathcal{S}(k)$ to the low-frequencies. This rapid collapse of energy occurs due to the alignment of dominant NTK eigenvectors with smooth, low-frequency functions, restricting expressivity in tasks requiring fine-scale resolution such as audio representation. The proposed WINNER used in SIREN$^2$ provides significant improvement, featuring a tunable $\mathcal{S}(k)$ profile controlled by weight perturbation scales $s_0$ and $s_1$, analogous to the Fourier scales used in random Fourier embeddings~\cite{rahimi2007random, tancik2020fourier}. By contrast, SIREN$^2$ captures high-frequencies and positional information directly through its sinusoidal activations, avoiding the quadratic increase in parameters required by random Fourier or positional embeddings to represent similar frequency content. As demonstrated in Fig.~\ref{fig:ntkspecs}, SIREN$^2$ can be configured to exhibit slower eigenvalue decay with higher cumulative spectral energy $\mathcal{S}(k)$ distributed across the entire spectrum with parameters $s_0$ and $s_1$. The critical factor for a successful implicit representation is the appropriate selection of $s_0$ and $s_1$ (Sec.~\ref{sec:specifyScales}).

\begin{figure}[t!]
    \centering
    \includegraphics[width=\linewidth]{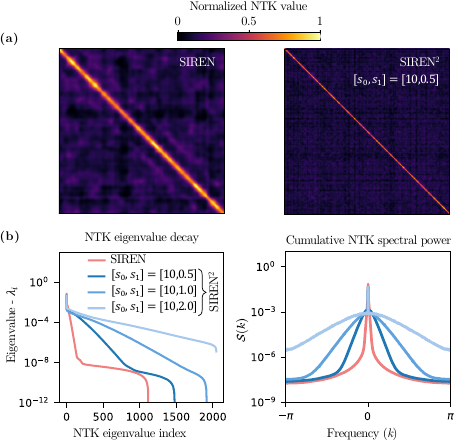}
    \caption{\textbf{Controlling the NTK spectra via noise scales $s_0$ and $s_1$ of WINNER.}  (a) Normalized NTK kernels for SIREN and SIREN$^2$ with $[s_0, s_1] = [10, 0.5]$, showing reduced off-diagonal correlations. (b) NTK eigenvalue decay profile (left) and eigenvalue-weighted FFT magnitude spectra - $\mathcal{S}(k)$ (right) for varying noise scales, demonstrating that SIREN$^2$ broadens frequency support. All networks shown here employ four hidden layers with 256 features, use $\omega_0=30$, and are evaluated on $2^{10}$ uniformly sampled inputs in $[-1, 1]$.}
    \label{fig:ntkspecs}
\end{figure}





\subsection{Activation Spectra}
Figure~\ref{fig:spectrum1} compares the distributions and power spectral densities (PSDs) of network inputs and the pre- and post-activation values in the first two hidden layers for SIREN and SIREN$^2$. For SIREN, pre-activation distributions are approximately Gaussian, while post-activations follow the $\text{Arcsin}(-1,1)$ law, consistent across layers. SIREN$^2$ preserves these distributional structure, albeit with broader pre-activation spreads in layers 1 and 2. Distributions for subsequent layers (not shown) remain identical to that of SIREN; full-layer distributions are reported in the Supplementary (Sec.~\ref{suppl:fulldists}).

Although similar distributions in the real space, differences arise in the spectral domain, especially the shape of spectra, due to the noise scales $s_0$ and $s_1$. SIREN exhibits dominant low-frequency content with negligible excitation of higher modes across layers, consistent with spectral bias. In contrast, SIREN$^2$ introduces structured broadband spectral energy right from the first hidden layer, which propagate through subsequent layers up to the output. This leads to a sustained high-frequency energy in the activations, enhancing the conditioning of the optimization landscape for regressing high-frequency signals. SIREN$^2$ retains the favorable bell type distributional properties of SIREN while enabling superior high-frequency receptivity at initialization.


\begin{figure*}[t!]
    \centering
    \includegraphics[width=0.9\linewidth]{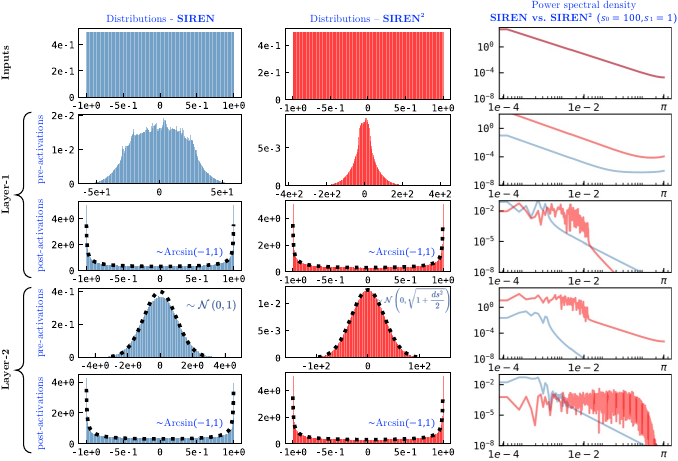}
    \caption{\textbf{WINNER enhances high-frequency receptivity.} Input, pre-activation, and post-activation distributions are shown alongside layer-averaged power spectral densities (PSDs) of SIREN and SIREN$^2$ up to layer-2 at initialization. SIREN$^2$ exhibits higher variance in its pre-activation distributions due to weight perturbations introduced by WINNER ($s_0=100, s_1=1$). The PSDs reveal that SIREN has limited high-frequency content, whereas SIREN$^2$ maintains broader spectral coverage with larger amplitudes at higher frequencies. This behavior persists through depth and extends to the outputs. A detailed analysis for a four-layer network is provided in the supplementary material. Both models use four hidden layers with 2048 hidden units per layer and are evaluated on $2^{10}$ uniformly spaced inputs over $\mathbf{x} \sim \mathcal{U}(-1,1)$ with $\omega_0=30$.}
    \label{fig:spectrum1}
\end{figure*}


\section{Experiments} \label{sec:experiments}
We evaluate the performance of SIREN$^2$ (SIREN initialized with WINNER) against several state-of-the-art INR architectures from the literature, including the baseline SIREN \cite{sitzmann2020implicit}, Gauss \cite{ramasinghe2022beyond} (2022), WIRE \cite{saragadam2023wire} (2023), and FINER \cite{liu2024finer} (2024), in reconstructing a variety of challenging audio signals. Following the recommendation in~\cite{sitzmann2020implicit}, all architectures, except SIREN$^2$, use a scaling factor of 100 for the first-layer activation periodicity for the audio fitting experiments.




\subsection{1D Audio Fitting} \label{sec:audio}
To ensure consistency across experiments, all audio signals are fixed at a length of 150{,}000 samples, making the experimental setup invariant to sampling rate.  Each network is trained for 30{,}000 epochs, and every experiment is conducted over five random trials to compute the mean peak PSNR and standard deviation. The selected audio signals span a range of spectral characteristics, including high-frequency-dominant, low-frequency-dominant, and broadband signals.

Table~\ref{tab:crit_fed} shows that the proposed SIREN$^2$ delivers consistently higher reconstruction accuracy than existing INR architectures across diverse audio signals, establishing new state-of-the-art results. These gains are especially notable for signals with strong high-frequency content, where other methods exhibit substantial residual errors. While the weight perturbation scheme amplifies high-frequency energy (Fig.~\ref{fig:dists}), it slightly suppresses low frequencies, as seen in the reconstruction of \texttt{bach.wav}. This indicates the need for a more robust choice of $s_0$ and $s_1$ for such low-frequency dominant signals. As illustrated in Fig.~\ref{fig:spectrograms}, SIREN$^2$ enables high-fidelity audio reconstructions, achieving PSNR values above $60$~dB with parameter count comparable to that of number of samples in the signal.



\begin{table*}[t!]
\centering
\caption{\textbf{Audio fitting.} Mean and standard deviation of PSNR values for different architectures on audio signal reconstruction. The proposed SIREN$^2$ achieves \textit{state-of-the-art performance}. Results are color coded as \colorbox{red!30}{best}, \colorbox{green!30}{second best}, and \colorbox{cyan!30}{third best} reconstructions. The network width is chosen so that the total parameter count is approximately equal to the signal length.}
\begin{tabular}{l @{\hspace{0.01cm}} >{\centering\arraybackslash}p{1.8cm} @{\hspace{0.12cm}} >{\centering\arraybackslash}p{1.8cm} @{\hspace{0.12cm}} >{\centering\arraybackslash}p{1.8cm} @{\hspace{0.12cm}} > {\centering\arraybackslash}p{1.8cm} @{\hspace{0.12cm}} >{\centering\arraybackslash}p{1.8cm} @{\hspace{0.12cm}} >{\centering\arraybackslash}p{1.8cm} @{\hspace{0.12cm}} >{\centering\arraybackslash}p{2.0cm}}
\toprule
& SIREN & FINER & WIRE & ReLU-PE & SIREN-RFF & FINER++ & \textbf{SIREN}$^\mathbf{2}$ \textbf{(present)} \\
\midrule
\rowcolor{gray!10}
Hidden layers & \texttt{4$\times$222} & \texttt{4$\times$222} & \texttt{4$\times$157} & \texttt{4$\times$193} & \texttt{4$\times$193} & \texttt{4$\times$222} & \texttt{4$\times$222} \\
\rowcolor{gray!10}
\# Fourier features & \texttt{0} & \texttt{0} & \texttt{0} & \texttt{193} & \texttt{193} & \texttt{0} & \texttt{0} \\
\rowcolor{gray!10}
\# parameters & \texttt{149185} & \texttt{149185} & \texttt{149474} & \texttt{150155} & \texttt{150155} & \texttt{149185} & \texttt{149185} \\
\midrule
\textbf{PSNR (dB)} ($\uparrow$):  & & & & & & & \\
\texttt{tetris.wav} & \texttt{13.4}$\pm$\texttt{0.0} & \texttt{13.6}$\pm$\texttt{0.0} & \texttt{13.6}$\pm$\texttt{0.0} & \texttt{13.6}$\pm$\texttt{0.0} & \colorbox{cyan!30}{\texttt{38.1}$\pm$\texttt{0.3}} & \colorbox{green!30}{\texttt{52.2}$\pm$\texttt{0.7}} & \colorbox{red!30}{\texttt{62.7}$\pm$\texttt{0.4}} \\
\texttt{tap.wav} & \texttt{20.4}$\pm$\texttt{0.0} & \texttt{21.1}$\pm$\texttt{0.0} & \texttt{21.1}$\pm$\texttt{0.0} & \texttt{21.1}$\pm$\texttt{0.0} & \colorbox{cyan!30}{\texttt{44.8}$\pm$\texttt{0.4}} & \colorbox{green!30}{\texttt{51.8}$\pm$\texttt{0.3}} & \colorbox{red!30}{\texttt{53.5}$\pm$\texttt{0.9}} \\
\texttt{whoosh.wav} & \texttt{33.8}$\pm$\texttt{0.9} & \colorbox{cyan!30}{\texttt{53.4}$\pm$\texttt{1.0}} & \texttt{20.2}$\pm$\texttt{0.0} & \texttt{20.2}$\pm$\texttt{0.0} & \texttt{41.8}$\pm$\texttt{0.6} & \colorbox{green!30}{\texttt{55.4}$\pm$\texttt{0.6}} & \colorbox{red!30}{\texttt{64.9}$\pm$\texttt{1.7}} \\
\texttt{radiation.wav} & \texttt{32.3}$\pm$\texttt{0.0} & \texttt{34.2}$\pm$\texttt{0.1} & \texttt{34.2}$\pm$\texttt{0.0} & \texttt{34.2}$\pm$\texttt{0.2} & \colorbox{green!30}{\texttt{52.4}$\pm$\texttt{0.1}} & \colorbox{cyan!30}{\texttt{50.9}$\pm$\texttt{1.8}} & \colorbox{red!30}{\texttt{63.0}$\pm$\texttt{1.0}} \\
\texttt{arch.wav} & \texttt{29.7}$\pm$\texttt{1.1} & \colorbox{cyan!30}{\texttt{58.5}$\pm$\texttt{0.8}} & \texttt{17.2}$\pm$\texttt{0.1} & \texttt{17.2}$\pm$\texttt{0.1} & \texttt{44.1}$\pm$\texttt{0.9} & \colorbox{green!30}{\texttt{65.2}$\pm$\texttt{0.2}} & \colorbox{red!30}{\texttt{95.2}$\pm$\texttt{2.9}} \\
\texttt{relay.wav} & \texttt{28.5}$\pm$\texttt{1.4} & \texttt{34.7}$\pm$\texttt{0.5} & \texttt{20.7}$\pm$\texttt{0.0} & \texttt{20.7}$\pm$\texttt{0.0} & \colorbox{cyan!30}{\texttt{40.5}$\pm$\texttt{0.6}} & \colorbox{green!30}{\texttt{54.1}$\pm$\texttt{0.4}} & \colorbox{red!30}{\texttt{60.4}$\pm$\texttt{2.9}} \\
\texttt{voltage.wav} & \texttt{34.0}$\pm$\texttt{0.8} & \colorbox{cyan!30}{\texttt{53.4}$\pm$\texttt{0.6}} & \texttt{20.0}$\pm$\texttt{0.0} & \texttt{19.9}$\pm$\texttt{0.0} & \texttt{43.7}$\pm$\texttt{0.3} & \colorbox{green!30}{\texttt{56.5}$\pm$\texttt{0.1}} & \colorbox{red!30}{\texttt{64.5}$\pm$\texttt{0.5}} \\
\texttt{foley.wav} & \texttt{36.6}$\pm$\texttt{7.2} & \colorbox{green!30}{\texttt{56.8}$\pm$\texttt{0.1}} & \texttt{29.7}$\pm$\texttt{0.1} & \texttt{22.5}$\pm$\texttt{0.0} & \texttt{44.9}$\pm$\texttt{0.3} & \colorbox{cyan!30}{\texttt{56.4}$\pm$\texttt{0.2}} & \colorbox{red!30}{\texttt{58.3}$\pm$\texttt{0.2}} \\
\texttt{shattered.wav} & \texttt{39.1}$\pm$\texttt{1.9} & \colorbox{green!30}{\texttt{58.6}$\pm$\texttt{0.4}} & \texttt{25.5}$\pm$\texttt{0.0} & \texttt{25.4}$\pm$\texttt{0.0} & \texttt{46.4}$\pm$\texttt{0.6} & \colorbox{cyan!30}{\texttt{57.9}$\pm$\texttt{0.3}} & \colorbox{red!30}{\texttt{64.7}$\pm$\texttt{0.7}} \\
\texttt{bach.wav} & \texttt{59.4}$\pm$\texttt{0.3} & \colorbox{red!30}{\texttt{64.5}$\pm$\texttt{0.2}} & \texttt{26.1}$\pm$\texttt{0.5} & \texttt{18.9}$\pm$\texttt{0.0} & \texttt{41.8}$\pm$\texttt{0.2} & \colorbox{green!30}{\texttt{62.2}$\pm$\texttt{0.3}} & \colorbox{cyan!30}{\texttt{60.5}$\pm$\texttt{0.2}} \\
\texttt{birds.wav} & \texttt{55.7}$\pm$\texttt{0.2} & \colorbox{green!30}{\texttt{59.6}$\pm$\texttt{0.1}} & \texttt{24.6}$\pm$\texttt{0.0} & \texttt{24.4}$\pm$\texttt{0.0} & \texttt{45.7}$\pm$\texttt{0.5} & \colorbox{cyan!30}{\texttt{58.7}$\pm$\texttt{0.1}} & \colorbox{red!30}{\texttt{61.2}$\pm$\texttt{0.2}} \\

\midrule
\texttt{Average} & \texttt{34.8}$\pm$\texttt{1.3} & \texttt{46.2}$\pm$\texttt{0.3} & \texttt{23.0}$\pm$\texttt{0.1} & \texttt{21.7}$\pm$\texttt{0.0} & \texttt{44.0}$\pm$\texttt{0.4} & \texttt{56.5}$\pm$\texttt{0.5} & \texttt{64.5}$\pm$\texttt{1.1} \\
\bottomrule
\end{tabular}
\label{tab:crit_fed}
\end{table*}

\begin{figure}[t!]
    \centering
    \includegraphics[width=\linewidth]{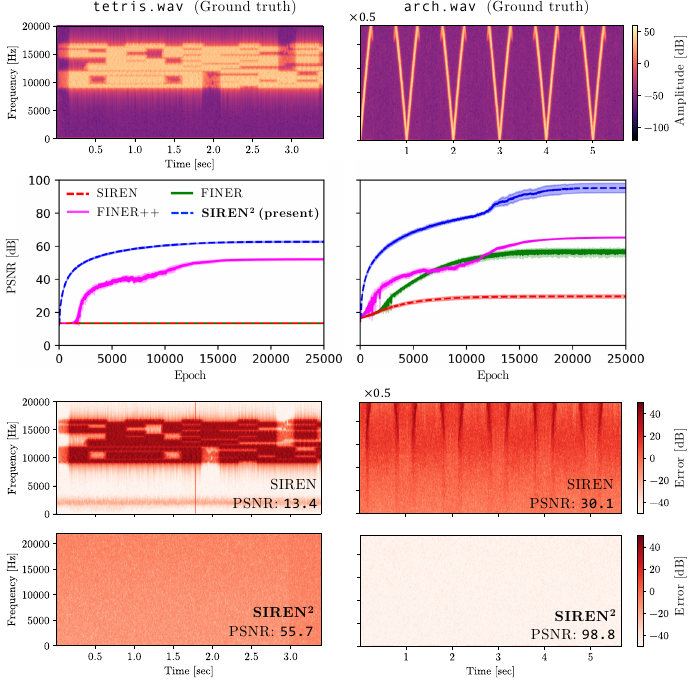}
    \caption{\textbf{High-fidelity audio reconstruction with SIREN$^2$.} Each column corresponds to an audio signal: \texttt{tetris.wav}, \texttt{sparking.wav}, and \texttt{shattered.wav}. Top row: ground truth spectrograms. Second row: Averaged PSNR histories from different network architectures with the region of standard deviation shaded. Remaining rows show the \textit{spectrogram errors} for reconstructions by each model. SIREN$^2$ stands out by exhibiting near-zero reconstruction noise.}
    \label{fig:spectrograms}
\end{figure}

\noindent \textbf{Reproducibility details.} The  inputs and audio targets are normalized to $[-1, 1]$ before training. SIREN and FINER use $\omega_0 = 30$, WIRE uses $\omega_0 = 10$ and $s_0 = 10$~\cite{saragadam2023wire}, and SIREN-RFF uses $\omega_0 = 30$ with Fourier embeddings drawn from $\mathcal{N}(0, 30^2)$. FINER++~\cite{zhu2024finerplusplus} employs a first-layer bias uniformly distributed in $[-5, 5]$. For all networks, the first-layer $\omega_0$ is scaled by $100$. Training uses a learning-rate scheduler that decays by 1\% every 20 epochs from an initial value of $10^{-4}$. Audio samples and code are provided in the linked GitHub repository.





\subsection{2D Image Fitting} \label{sec:images}
We evaluate SIREN$^2$ on 2D image fitting tasks, $f(\mathbf{x};\theta): \mathbb{R}^2 \mapsto \mathbb{R}^d$, with $d=1$ for grayscale and $d=3$ for RGB images. The experiments cover a diverse set of images, including natural images from the Kodak~\cite{kodak_photocd} dataset, challenging texture images from two DTD~\cite{cimpoi14describing} classes (\emph{braided} and \emph{woven}), and synthetic high-frequency patterns. Table~\ref{tab:images1} shows that SIREN$^2$ consistently outperforms the original SIREN across all cases, with PSNR improvements ranging from 7\% (\texttt{2D-Riemann}) to 69\% (\texttt{noise.png}), and especially strong gains for images with high-frequency content and small pixel count (in the over-paramaterized regime). For RGB image reconstruction, SIREN$^2$ provides only marginal gains over SIREN, while FINER achieves the best performance. Figure~\ref{fig:image_exps} visualizes these improvements using FFT error maps, highlighting that SIREN$^2$ achieves lower fitting error in high-frequency regions. The frequency-domain analysis further confirms that SIREN$^2$ preserves fine details and sharp transitions more accurately, as indicated by the reduced magnitude errors (darker regions).

\begin{table*}[h!]
\centering
\caption{\textbf{2D image fitting.} PSNR($\uparrow$) in dB across images and datasets for different INR architectures. SIREN$^2$ consistently surpasses SIREN, with percentage gains (in parentheses) reflecting improvements achieved purely through initialization. Results are color coded as \colorbox{red!30}{best}, \colorbox{green!30}{second best}, and \colorbox{cyan!30}{third best} reconstructions.}
\begin{tabular}{l @{\hspace{0.01cm}} >{\centering\arraybackslash}p{1.4cm} @{\hspace{0.12cm}} >{\centering\arraybackslash}p{2.2cm} @{\hspace{0.12cm}} >{\centering\arraybackslash}p{1.4cm} @{\hspace{0.12cm}} > {\centering\arraybackslash}p{1.4cm} @{\hspace{0.12cm}} > {\centering\arraybackslash}p{1.4cm} @{\hspace{0.12cm}} >{\centering\arraybackslash}p{1.4cm}}
\toprule
& SIREN & \textbf{SIREN$^2$ (present)} & ReLU-PE & WIRE & FINER & Gauss \\
\midrule
\rowcolor{gray!10}
Hidden layers ($n\times w$) & \texttt{4$\times$256} & \texttt{4$\times$256} & \texttt{4$\times$256} & \texttt{4$\times$128} & \texttt{4$\times$256} & \texttt{4$\times$256} \\
\rowcolor{gray!10}
\# parameters & \texttt{198145} & \texttt{198145} & \texttt{263553} & \texttt{198386} & \texttt{198145} & \texttt{198145}  \\
\midrule
\textbf{Peak PSNR (dB)} ($\uparrow$): & & & & & &  \\
noise.png & 21.3 & \colorbox{red!30}{\textbf{36.1}} (\textcolor{teal}{69\% $\uparrow$}) & 16.9 & 25.5 & \colorbox{cyan!30}{33.0} & \colorbox{green!30}{34.1} \\
camera.png & \colorbox{cyan!30}{38.9} & \colorbox{green!30}{\textbf{44.9}} (\textcolor{teal}{15\% $\uparrow$}) & 28.4 & 37.2 & \colorbox{red!30}{46.4} & 28.6 \\
castle.jpg & \colorbox{cyan!30}{33.6} & \colorbox{green!30}{\textbf{36.5}} (\textcolor{teal}{9\% $\uparrow$}) & 22.3 & 28.5 & \colorbox{red!30}{36.9} & 19.2 \\
rock.png & 26.9 & \colorbox{green!30}{\textbf{36.2}} (\textcolor{teal}{35\% $\uparrow$}) & 16.1 & 26.5 & \colorbox{red!30}{36.6} & \colorbox{cyan!30}{31.8} \\
2D-Riemann (CFD data) & \colorbox{cyan!30}{55.3} & \colorbox{green!30}{\textbf{59.1}} (\textcolor{teal}{7\% $\uparrow$}) & 45.8 & 49.7 & \colorbox{red!30}{60.5} & 28.1 \\
DTD braided dataset (120 images, gray mode)$^*$ & \colorbox{cyan!30}{48.6} & \colorbox{red!30}{\textbf{75.2}} (\textcolor{teal}{55\% $\uparrow$}) & - & - & \colorbox{green!30}{65.4} & - \\
DTD woven dataset (120 images, gray mode)$^*$ & \colorbox{cyan!30}{41.9} & \colorbox{red!30}{\textbf{61.0}} (\textcolor{teal}{46\% $\uparrow$}) & - & - & \colorbox{green!30}{53.1} & - \\
    Kodak dataset (24 images, gray mode)$^*$ & \colorbox{cyan!30}{34.9} & \colorbox{green!30}{\textbf{37.6}} (\textcolor{teal}{8\% $\uparrow$}) & - & - & \colorbox{red!30}{38.1} & - \\
\bottomrule
\multicolumn{7}{r}{\footnotesize $^*$The reported PSNR values for these datasets represent averages computed over all images.}
\end{tabular}
\label{tab:images1}
\end{table*}

\begin{figure}
    \centering
    \includegraphics[width=\linewidth]{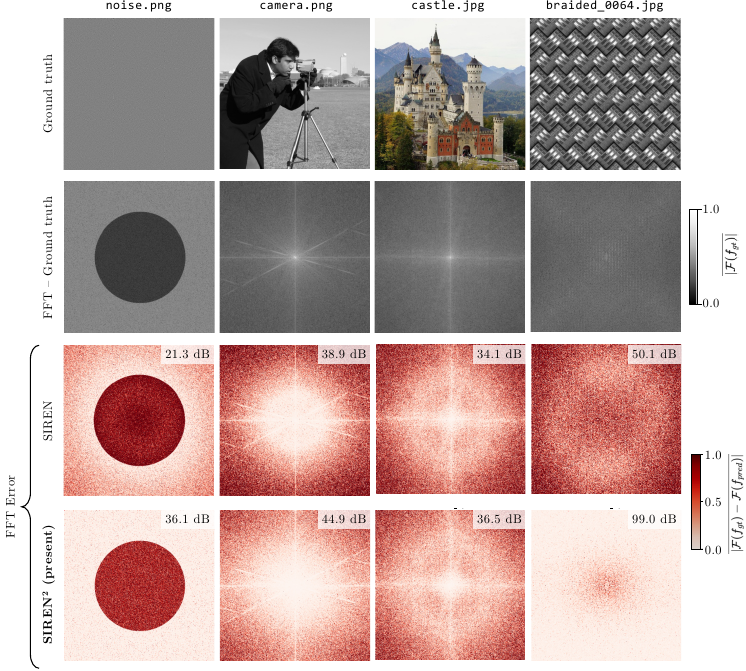}
    \caption{\textbf{Accuracy improvements in Fourier space with SIREN$^2$ for image fitting.} Reconstruction performance of SIREN and \textsc{SIREN\textsuperscript{2}} models for various images. The top two rows show the ground truth images and their corresponding FFT magnitudes ($|\mathcal{F}(f_{gt})|$). Subsequent rows depict the FFT error maps and peak PSNR values achieved by different models. \textsc{SIREN\textsuperscript{2}} consistently produce lower FFT errors and higher PSNRs.}
    \label{fig:image_exps}
\end{figure}

\textbf{Reproducibility details.} All networks follow the same settings as the audio fitting experiments, with no scaling applied to the first layer $\omega_0$. The ReLU+PE presented in Table~\ref{tab:images1} incorporates positional encoding with 256 embeddings. The positional encodings use a logarithmic frequency spectrum, with frequencies ranging from $2^0$ to $2^{\text{n}-1}$, with n=7 (number of frequencies).

\subsection{Image Denoising} \label{sec:imageDeNoise}
\begin{figure}[t!]
    \vspace{-10pt}
    \includegraphics[width=\linewidth]{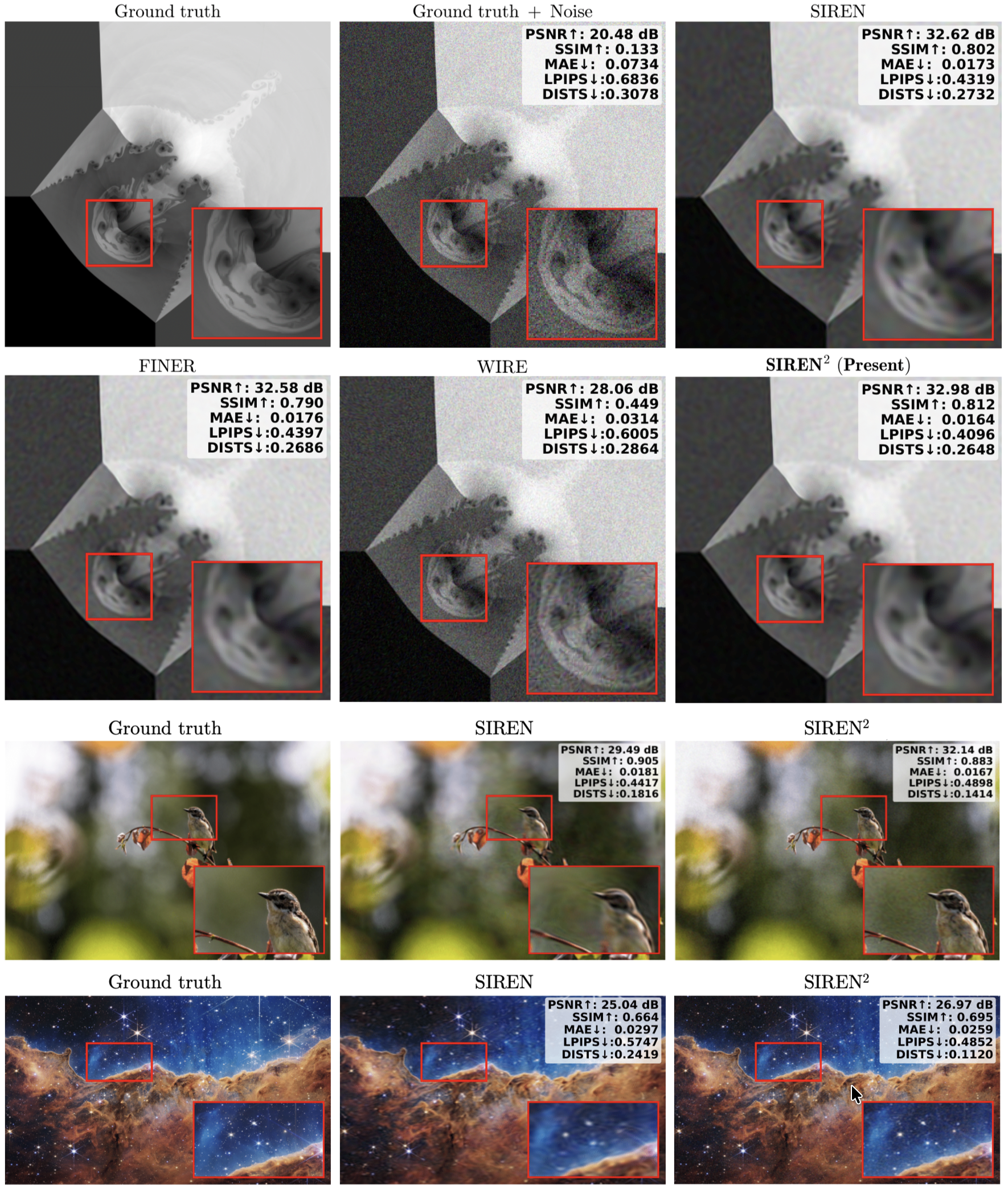}
    \caption{\textbf{Image denoising.} Comparison of reconstructions for noisy 2D fields and natural images. Rows 1–2: reconstruction of a noisy density field obtained from a computational fluid dynamics simulation of a 2D Riemann problem \cite{kurganov2002solution,chandravamsi2024high}. Rows 3–4: reconstruction of a sparrow image and a nebula image. SIREN$^2$ yields the highest fidelity among the compared methods, with quantitative results summarized in Table~\ref{tab:image-denoise}.}
    \label{fig:image-denoise}
\end{figure}

The robustness of different INR architectures is assessed for the canonical image denoising. 
Firstly, a clean signal $f(\mathbf{x})$ is corrupted by additive white Gaussian noise $\eta(\mathbf{x}) \sim \mathcal{N}(0,\sigma^2)$ such that 
$\tilde{f}(\mathbf{x}) = f(\mathbf{x}) + \eta(\mathbf{x})$ achieves a signal-to-noise ratio of $\mathrm{SNR}=5\ \mathrm{dB}$. 
The task is to reconstruct $f$ from $\tilde{f}$. 
We adopt an unsupervised denoising strategy similar to Noise2Self~\cite{batson2019noise2self}, training the INR directly on the noisy input while reserving a small subset of pixels for `J-invariant' validation. A predictor is J-invariant if, for any pixel $i$, the prediction at $i$ does not depend on the noisy value at $i$ itself. This restriction forces the model to reconstruct structure from spatial context rather than memorize pixelwise noise. The held-out set then provides an unsupervised early-stopping criterion that prevents overfitting to noise while retaining underlying structure. \textit{Ground-truth images are used strictly for evaluation (e.g., PSNR, SSIM) and never for model selection.} For the image quality evaluation, along with PSNR we report structural similarity~\cite{wang2004image},  mean absolute error (MAE), LPIPS~\cite{zhang2018unreasonable}, and DISTS~\cite{ding2020image}. For clarity, arrows ($\uparrow$ / $\downarrow$) are used in tables and figures to indicate whether higher or lower values correspond to better performance.

\begin{table}[t!]
\centering
\caption{Quantitative comparison of different architectures on noisy image reconstruction. SIREN$^2$ (ours) consistently achieves the best or comparable scores across all metrics, showing better noise suppression and structural fidelity.}
\renewcommand{\arraystretch}{0.95}
\setlength{\tabcolsep}{0.1pt}
\begin{tabular}{
  l 
  @{\hspace{0.01cm}} >{\centering\arraybackslash}p{1.2cm}
  @{\hspace{0.01cm}} >{\centering\arraybackslash}p{1.2cm}
  @{\hspace{0.01cm}} >{\centering\arraybackslash}p{1.2cm}
  @{\hspace{0.01cm}} >{\centering\arraybackslash}p{1.2cm}
  @{\hspace{0.01cm}} >{\centering\arraybackslash}p{1.25cm}
}
\toprule
 & PSNR($\uparrow$) & SSIM($\uparrow$) & MAE($\downarrow$) & LPIPS($\downarrow$) & DISTS($\downarrow$) \\
\midrule
\rowcolor{gray!10} 
\textbf{2D Riemann (Fig.~\ref{fig:image-denoise})} & & & & & \\
Ground truth + Noise & 20.48 & 0.133 & 0.0734 & 0.6836 & 0.3078 \\
SIREN & \colorbox{green!30}{32.62} & \colorbox{green!30}{0.802} & \colorbox{green!30}{0.0173} & \colorbox{green!30}{0.4139} & \colorbox{cyan!30}{0.2732} \\
FINER & \colorbox{cyan!30}{32.58} & \colorbox{cyan!30}{0.790} & \colorbox{cyan!30}{0.0176} & \colorbox{cyan!30}{0.4397} & \colorbox{green!30}{0.2686} \\
WIRE & 28.06 & 0.449 & 0.0314 & 0.6005 & 0.2864 \\
SIREN$^2$ (Ours) & \colorbox{red!30}{32.98} & \colorbox{red!30}{0.812} & \colorbox{red!30}{0.0164} & \colorbox{red!30}{0.4096} & \colorbox{red!30}{0.2648} \\
\midrule
\rowcolor{gray!10}
\textbf{Sparrow} & & & & & \\
Ground truth + Noise & 20.44 & 0.144 & 0.0753 & 0.7077 & 0.2722 \\
SIREN & 29.49 & \colorbox{red!30}{0.905} & \colorbox{cyan!30}{0.0181} & \colorbox{red!30}{0.4417} & 0.1816 \\
FINER & \colorbox{green!30}{31.90} & \colorbox{cyan!30}{0.880} & \colorbox{green!30}{0.0172} & \colorbox{cyan!30}{0.4929} & \colorbox{green!30}{0.1423} \\
WIRE & \colorbox{cyan!30}{30.23} & 0.816 & 0.0205 & 0.5336 & \colorbox{cyan!30}{0.1792} \\
SIREN$^2$ (Ours) & \colorbox{red!30}{32.14} & \colorbox{green!30}{0.883} & \colorbox{red!30}{0.0167} & \colorbox{green!30}{0.4899} & \colorbox{red!30}{0.1414} \\
\midrule
\rowcolor{gray!10}
\textbf{Nebula} & & & & & \\
Ground truth + Noise & 20.31 & 0.261 & 0.0775 & 0.5554 & 0.2196 \\
SIREN & \colorbox{cyan!30}{25.04} & \colorbox{cyan!30}{0.664} & \colorbox{cyan!30}{0.0297} & 0.5747 & 0.2419 \\
FINER & \colorbox{green!30}{26.37} & \colorbox{red!30}{0.697} & \colorbox{green!30}{0.0263} & \colorbox{green!30}{0.5055} & \colorbox{green!30}{0.1503} \\
WIRE & 25.02 & 0.582 & 0.0350 & \colorbox{cyan!30}{0.5410} & \colorbox{cyan!30}{0.1997} \\
SIREN$^2$ (Ours) & \colorbox{red!30}{26.97} & \colorbox{green!30}{0.695} & \colorbox{red!30}{0.0259} & \colorbox{red!30}{0.4827} & \colorbox{red!30}{0.1120} \\
\bottomrule
\end{tabular}
\label{tab:image-denoise}
\end{table}

Table~\ref{tab:image-denoise} summarizes the quantitative evaluation of image denoising performance across different models using standard metrics, including PSNR, SSIM, MAE, LPIPS, and DISTS. The results demonstrate that SIREN$^2$ consistently achieves higher fidelity reconstructions with improved perceptual and structural quality compared to the baseline SIREN and other methods. Snapshots of reconstructions presented in Fig.~\ref{fig:image-denoise} with regard to the 2D-Riemann problem suggest that SIREN and FINER oversmooth fine-scale details, while WIRE preserves global appearance but leaves residual noise. SIREN$^2$ achieves the sharpest and most faithful reconstructions (with the best SSIM and LPIPS), suppressing noise while retaining textures and edges.

\textbf{Reproducibility details.} The ground truth signal is corrupted with Gaussian noise at an SNR of 5 dB. All experiments employ a four-layer architecture with 256 features per layer. The search space for SIREN$^2$ noise scales is confined to $s_0 \in [0,200]$ with a fixed $s_1=0.01$. For FINER++, the bias scale $k$ in $\bar{b}\sim \mathcal{U}(-k,k)$ of the first hidden layer is set to $k \in [0,20]$. The first-layer activation periodicity $\omega_0$ is kept at 30 in both FINER++ and SIREN$^2$, ensuring comparable spectral bias.

\subsection{Audio Denoising} \label{sec:audioDeNoise}

\begin{table}
\centering
\caption{\textbf{Audio denoising.} Best PSNR ($\uparrow$) in dB for different audio clips using FINER++ and SIREN$^2$.}
\begin{tabular}{l 
  @{\hspace{0.01cm}} >{\centering\arraybackslash}p{1.6cm}
  @{\hspace{0.12cm}} >{\centering\arraybackslash}p{1.4cm} 
  @{\hspace{0.12cm}} >{\centering\arraybackslash}p{2.5cm} }
\toprule
 & GT+Noise & FINER++ & \textbf{SIREN$^2$ (present)} \\
\midrule
\texttt{bach.wav}      & 21.16 & 34.69 & \textbf{35.53} \\
\texttt{dilse.wav}     & 21.46 & 35.16 & \textbf{35.18} \\
\texttt{birds.wav}     & 29.92 & 34.84 & \textbf{37.17} \\
\texttt{counting.wav}  & 26.67 & 38.19 & \textbf{38.71} \\
\bottomrule
\end{tabular}
\label{tab:audio-denoise}
\end{table}

We adopt the same Noise2Self-inspired DIP training procedure from Sec.~\ref{sec:imageDeNoise} for the present audio denoising experiments. Ground-truth audio signal is used only for evaluation. While the general denoising framework is identical, audio signals present distinct challenges. Unlike natural images, which concentrate most energy in low frequencies, audio signals often exhibit a relatively broadband structure (e.g., higher harmonics in music, ambient noise), causing stronger overlap between the underlying signal and Gaussian noise. This overlap makes simple frequency-selective filtering less effective. In an a-priori setting, where no ground-truth information is available, INR-based denoising can provide a useful alternative.

in Sec.~\ref{sec:audio} on supervised audio fitting task, $\omega_0$ was scaled by a factor of 100 consistently for all network architectures. However, for the present audio-denoising task, a large first-layer periodicity consistently degraded denoising for SIREN, Gauss, WIRE, and FINER, with best denoising $\mathrm{PSNR} < 25\,\mathrm{dB}$ for different signals; due to the presence of broadband noise in the noisy signal. We therefore avoid first-layer $\omega_0$ scaling for denoising and report results for FINER++ and SIREN$^2$ architectures only which does not use any first layer $\omega_0$ scaling. Table~\ref{tab:audio-denoise} and Fig.~\ref{fig:audio-denoise} report the accuracy of denoised reconstructions for \texttt{bach.wav} audio clip after 20,000 epochs for FINER++ and SIREN$^2$. The noise scales $s_0$ and $s_1$ in SIREN$^2$ provide controllable filtering, allowing a principled adjustment of frequency support to match the noise characteristics. For FINER++, varying the bias ranges was evaluated. The results in Table~\ref{tab:audio-denoise} show that both networks achieve competitive accuracy in the recovery of the underlying signal.

\textbf{Reproducibility details.} The ground truth signal is corrupted with Gaussian noise at an SNR of 5 dB. All experiments employ a four-layer architecture with 256 features per layer. The search space for SIREN$^2$ noise scales is confined to $s_0 \in [800,2000]$ with a fixed $s_1=0.001$. For FINER++, the bias scale $k$ in $\bar{b}\sim \mathcal{U}(-k,k)$ of the first hidden layer is set to $k \in [0,20]$. The first-layer activation periodicity $\omega_0$ is kept at 30 in both FINER++ and SIREN$^2$. No scaling factor was used to increase the $\omega_0$ value for the first layer activations. 

\begin{figure*}
    \centering
    \includegraphics[width=\linewidth]{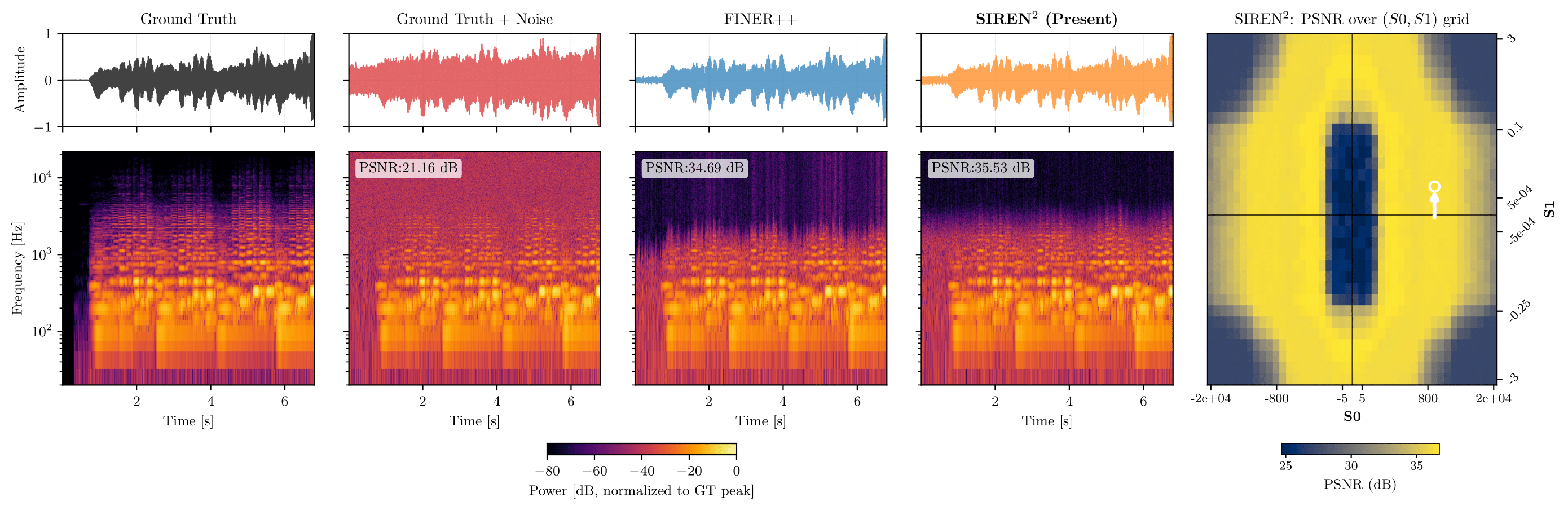}
    \caption{\textbf{Audio denoising.} Log-spectrogram comparison for audio denoising task performed on \texttt{bach.wav} using FINER++ and SIREN$^2$. SIREN$^2$ achieves higher PSNR and preserves more harmonic detail, as visible in the spectrogram. The rightmost plot shows the PSNR variation of SIREN$^2$ over the $(s_0, s_1)$ grid, illustrating that the noise scales $s_0$ and $s_1$ provide flexibility to control the degree of denoising through modulation of high-frequency suppression.}
    \label{fig:audio-denoise}
\end{figure*}


\subsection{Audio Inpainting} \label{sec:inpaint}
We evaluate the inpainting capability of SIREN$^{2}$ using the same network width, signal length, and noise scales employed in the audio fitting experiments of Sec.~\ref{sec:audio}. A fixed percentage of sample values is removed and the model is trained only on the observed subset. Sparsity denotes the retained fraction of samples in the training set. No spatial or structural priors are used; the task reduces to regressing the underlying continuous audio data from incomplete pointwise samples. 

Across all sparsity levels, SIREN$^{2}$ converges to lower reconstruction error for most cases than the baseline SIREN. The gains are consistent for both low and moderate sparsity, indicating that the improved high-frequency receptivity at initialization can benefit recovery of missing regions. The reconstructed samples exhibit sharper transitions indicating influence of initialization scheme on the INR’s ability to approximate the fine-scale details under limited observations.

\begin{table}[!ht]
    \centering
    \caption{\textbf{Audio inpainting.} Mean and standard deviation of PSNR (in dB) for different sparsity levels using SIREN and SIREN$^{2}$.}
    \begin{tabular}{lccc}
        \hline
        \textbf{Audio sample} & \textbf{Sparsity} & \textbf{SIREN} & \textbf{SIREN$^{2}$}\\
        \hline
        \multirow{4}{*}{bach.wav}
         & 0.05 & 18.87 $\pm$ 0.28 & 19.19 $\pm$ 0.03 \\
         & 0.1  & 22.26 $\pm$ 0.14 & 22.71 $\pm$ 0.13 \\
         & 0.25 & 32.79 $\pm$ 0.63 & 32.97 $\pm$ 0.37 \\
         & 0.5  & 46.89 $\pm$ 0.33 & 46.81 $\pm$ 0.44 \\
        \hline
        \multirow{4}{*}{relay.wav}
         & 0.05 & 16.53 $\pm$ 1.51 & 18.34 $\pm$ 0.08\\
         & 0.1  & 15.92 $\pm$ 0.28 & 18.10 $\pm$ 0.17\\
         & 0.25 & 16.79 $\pm$ 0.07 & 18.18 $\pm$ 0.18\\
         & 0.5  & 17.77 $\pm$ 0.03 & 19.18 $\pm$ 0.16\\
        \hline
        \multirow{4}{*}{tetris.wav}
         & 0.05 & 8.80  $\pm$ 0.13 & 12.19 $\pm$ 0.02\\
         & 0.1  & 9.21  $\pm$ 0.07 & 11.95 $\pm$ 0.03\\
         & 0.25 & 10.13 $\pm$ 0.21 & 11.87 $\pm$ 0.08\\
         & 0.5  & 12.19 $\pm$ 0.27 & 12.50 $\pm$ 0.08\\
        \hline
    \end{tabular}
\end{table}

\section{Conclusion} \label{sec:conclusions}

Deep neural networks exhibit a spectral learning bias, where low-frequency components are learned early in training while higher-frequency modes emerge more gradually in the later epochs. However, when the target signal lacks low-frequency content and is dominated by high frequencies, training can suffer from a spectral bottleneck, leading to failure in reconstructing the signal. The present work examined such a phenomenon in the context of sinusoidal representation networks (SIRENs), with a focus on fitting signals dominated by high-frequency content. We analyze the evolution of activation spectra and the neural tangent kernel and characterize the mechanisms underlying the spectral bottleneck. To address this limitation, we propose a target-aware Gaussian weight perturbation scheme, WINNER. The method applies Gaussian noise of a specified scale to the uniformly initialized weights, thereby modifying pre-activation distributions and layerwise spectra. Similar to random Fourier embeddings \cite{tancik2020fourier}, the present weight perturbation scheme can be used to control the empirical NTK and its eigenbasis, with the added benefit of not introducing additional trainable parameters. The approach achieves state-of-the-art performance on audio fitting tasks and highlights the critical role of initial activation spectra in determining fitting accuracy.

\noindent \textbf{Limitations.} (a) The proposed initialization scheme relies on prior knowledge or an approximate estimate of the target’s spectral content to determine the perturbation scales $s_0$ and $s_1$. For problems where the target is not known in advance, such as denoising or solving initial/boundary value PDEs, these hyperparameters can be initialized with estimated values and adaptively updated during early training. A similar limitation also arises when specifying the Fourier scale for other equivalent methods such as random Fourier feature embeddings \cite{tancik2020fourier}. (b) Similar to the other architectures examined in this study, the training dynamics and final performance of SIREN$^2$ are sensitive to the learning rate and its decay schedule. Without a scheduler, training often results in a strongly oscillatory PSNR evolution resulting in poor convergence. We recommend a decay rate of 1-2\% every 20 epochs with an initial learning rate of $10^{-4}$ for both audio and image fitting to achieve stable, non-oscillatory PSNR evolution.


\section*{Acknowledgments}
The authors HC, DVS, IZ, ZC, and SP gratefully acknowledge the financial assistance provided by Technion - Israel Institute of Technology during the course of present work. HC thanks Sanketh Vedula (Princeton University) for directing him to the work of Sitzmann et al. \cite{sitzmann2020implicit}; experimenting with it helped the development of noise scheme presented in this work. Generative AI tools were used solely for language editing, limited to grammar and phrasing refinement, and were not used for developing or modifying technical content.

\bibliographystyle{unsrt}
\bibliography{ML_bibilography}

@article{mildenhall2021nerf,
  title={Nerf: Representing scenes as neural radiance fields for view synthesis},
  author={Mildenhall, Ben and Srinivasan, Pratul P and Tancik, Matthew and Barron, Jonathan T and Ramamoorthi, Ravi and Ng, Ren},
  journal={Communications of the ACM},
  volume={65},
  number={1},
  pages={99--106},
  year={2021},
  publisher={ACM New York, NY, USA}
}

@article{sitzmann2020implicit,
  title={Implicit neural representations with periodic activation functions},
  author={Sitzmann, Vincent and Martel, Julien and Bergman, Alexander and Lindell, David and Wetzstein, Gordon},
  journal={Advances in neural information processing systems},
  volume={33},
  pages={7462--7473},
  year={2020}
}

@article{tancik2020fourier,
  title={Fourier features let networks learn high frequency functions in low dimensional domains},
  author={Tancik, Matthew and Srinivasan, Pratul and Mildenhall, Ben and Fridovich-Keil, Sara and Raghavan, Nithin and Singhal, Utkarsh and Ramamoorthi, Ravi and Barron, Jonathan and Ng, Ren},
  journal={Advances in neural information processing systems},
  volume={33},
  pages={7537--7547},
  year={2020}
}

@article{sharan2021benchmarking,
  title={Benchmarking audio signal representation techniques for classification with convolutional neural networks},
  author={Sharan, Roneel V and Xiong, Hao and Berkovsky, Shlomo},
  journal={Sensors},
  volume={21},
  number={10},
  pages={3434},
  year={2021},
  publisher={MDPI}
}

@article{tessarini2022audio,
  title={Audio signals and artificial neural networks for classification of plastic resins for recycling},
  author={Tessarini, Let{\'\i}cia and Fileti, Ana Maria Frattini},
  journal={Digital Chemical Engineering},
  volume={5},
  pages={100059},
  year={2022},
  publisher={Elsevier}
}

@inproceedings{aironi2021graph,
  title={Graph-based representation of audio signals for sound event classification},
  author={Aironi, Carlo and Cornell, Samuele and Principi, Emanuele and Squartini, Stefano},
  booktitle={2021 29th European Signal Processing Conference (EUSIPCO)},
  pages={566--570},
  year={2021},
  organization={IEEE}
}

@inproceedings{hershey2017cnn,
  title={CNN architectures for large-scale audio classification},
  author={Hershey, Shawn and Chaudhuri, Sourish and Ellis, Daniel PW and Gemmeke, Jort F and Jansen, Aren and Moore, R Channing and Plakal, Manoj and Platt, Devin and Saurous, Rif A and Seybold, Bryan and others},
  booktitle={2017 ieee international conference on acoustics, speech and signal processing (icassp)},
  pages={131--135},
  year={2017},
  organization={IEEE}
}

@article{donahue2018adversarial,
  title={Adversarial audio synthesis},
  author={Donahue, Chris and McAuley, Julian and Puckette, Miller},
  journal={arXiv preprint arXiv:1802.04208},
  year={2018}
}

@inproceedings{shen2018natural,
  title={Natural tts synthesis by conditioning wavenet on mel spectrogram predictions},
  author={Shen, Jonathan and Pang, Ruoming and Weiss, Ron J and Schuster, Mike and Jaitly, Navdeep and Yang, Zongheng and Chen, Zhifeng and Zhang, Yu and Wang, Yuxuan and Skerrv-Ryan, Rj and others},
  booktitle={2018 IEEE international conference on acoustics, speech and signal processing (ICASSP)},
  pages={4779--4783},
  year={2018},
  organization={IEEE}
}

@article{jacot2018neural,
  title={Neural tangent kernel: Convergence and generalization in neural networks},
  author={Jacot, Arthur and Gabriel, Franck and Hongler, Cl{\'e}ment},
  journal={Advances in neural information processing systems},
  volume={31},
  year={2018}
}

@article{song2022versatile,
  title={A versatile framework to solve the Helmholtz equation using physics-informed neural networks},
  author={Song, Chao and Alkhalifah, Tariq and Waheed, Umair Bin},
  journal={Geophysical Journal International},
  volume={228},
  number={3},
  pages={1750--1762},
  year={2022},
  publisher={Oxford University Press}
}

@article{wang2022and,
  title={When and why PINNs fail to train: A neural tangent kernel perspective},
  author={Wang, Sifan and Yu, Xinling and Perdikaris, Paris},
  journal={Journal of Computational Physics},
  volume={449},
  pages={110768},
  year={2022},
  publisher={Elsevier}
}

@article{ziyin2020neural,
  title={Neural networks fail to learn periodic functions and how to fix it},
  author={Ziyin, Liu and Hartwig, Tilman and Ueda, Masahito},
  journal={Advances in Neural Information Processing Systems},
  volume={33},
  pages={1583--1594},
  year={2020}
}

@inproceedings{liu2013fourier,
  title={Fourier neural network for machine learning},
  author={Liu, Shuang},
  booktitle={2013 international conference on machine learning and cybernetics},
  volume={1},
  pages={285--290},
  year={2013},
  organization={IEEE}
}

@inproceedings{silvescu1999fourier,
  title={Fourier neural networks},
  author={Silvescu, Adrian},
  booktitle={IJCNN'99. International Joint Conference on Neural Networks. Proceedings (Cat. No. 99CH36339)},
  volume={1},
  pages={488--491},
  year={1999},
  organization={IEEE}
}

@article{lanzendorfer2023siamese,
  title={Siamese siren: Audio compression with implicit neural representations},
  author={Lanzend{\"o}rfer, Luca A and Wattenhofer, Roger},
  journal={arXiv preprint arXiv:2306.12957},
  year={2023}
}

@article{serrano2024hosc,
  title={HOSC: A Periodic Activation Function for Preserving Sharp Features in Implicit Neural Representations},
  author={Serrano, Danzel and Szymkowiak, Jakub and Musialski, Przemyslaw},
  journal={arXiv preprint arXiv:2401.10967},
  year={2024}
}

@article{rahimi2007random,
  title={Random features for large-scale kernel machines},
  author={Rahimi, Ali and Recht, Benjamin},
  journal={Advances in neural information processing systems},
  volume={20},
  year={2007}
}

@article{raissi2019physics,
  title={Physics-informed neural networks: A deep learning framework for solving forward and inverse problems involving nonlinear partial differential equations},
  author={Raissi, Maziar and Perdikaris, Paris and Karniadakis, George E},
  journal={Journal of Computational physics},
  volume={378},
  pages={686--707},
  year={2019},
  publisher={Elsevier}
}

@inproceedings{oechsle2019texture,
  title={Texture fields: Learning texture representations in function space},
  author={Oechsle, Michael and Mescheder, Lars and Niemeyer, Michael and Strauss, Thilo and Geiger, Andreas},
  booktitle={Proceedings of the IEEE/CVF international conference on computer vision},
  pages={4531--4540},
  year={2019}
}

@inproceedings{niemeyer2020differentiable,
  title={Differentiable volumetric rendering: Learning implicit 3d representations without 3d supervision},
  author={Niemeyer, Michael and Mescheder, Lars and Oechsle, Michael and Geiger, Andreas},
  booktitle={Proceedings of the IEEE/CVF conference on computer vision and pattern recognition},
  pages={3504--3515},
  year={2020}
}

@inproceedings{rahaman2019spectral,
  title={On the spectral bias of neural networks},
  author={Rahaman, Nasim and Baratin, Aristide and Arpit, Devansh and Draxler, Felix and Lin, Min and Hamprecht, Fred and Bengio, Yoshua and Courville, Aaron},
  booktitle={International conference on machine learning},
  pages={5301--5310},
  year={2019},
  organization={PMLR}
}

@article{xu2019frequency,
  title={Frequency principle: Fourier analysis sheds light on deep neural networks},
  author={Xu, Zhi-Qin John and Zhang, Yaoyu and Luo, Tao and Xiao, Yanyang and Ma, Zheng},
  journal={arXiv preprint arXiv:1901.06523},
  year={2019}
}

@inproceedings{basri2020frequency,
  title={Frequency bias in neural networks for input of non-uniform density},
  author={Basri, Ronen and Galun, Meirav and Geifman, Amnon and Jacobs, David and Kasten, Yoni and Kritchman, Shira},
  booktitle={International conference on machine learning},
  pages={685--694},
  year={2020},
  organization={PMLR}
}

@inproceedings{lee2019wide,
  author    = {Lee, Jaehoon and Xiao, Lechao and Schoenholz, Samuel and Bahri, Yasaman and Novak, Roman and Sohl-Dickstein, Jascha and Pennington, Jeffrey},
  title     = {Wide Neural Networks of Any Depth Evolve as Linear Models Under Gradient Descent},
  booktitle = {Advances in Neural Information Processing Systems},
  volume    = {32},
  year      = {2019}
}

@article{szatkowski2022hypersound,
  title={Hypersound: Generating implicit neural representations of audio signals with hypernetworks},
  author={Szatkowski, Filip and Piczak, Karol J and Spurek, Przemys{\l}aw and Tabor, Jacek and Trzci{\'n}ski, Tomasz},
  journal={arXiv preprint arXiv:2211.01839},
  year={2022}
}

@article{su2022inras,
  title={Inras: Implicit neural representation for audio scenes},
  author={Su, Kun and Chen, Mingfei and Shlizerman, Eli},
  journal={Advances in Neural Information Processing Systems},
  volume={35},
  pages={8144--8158},
  year={2022}
}

@inproceedings{mehta2021modulated,
  title={Modulated periodic activations for generalizable local functional representations},
  author={Mehta, Ishit and Gharbi, Micha{\"e}l and Barnes, Connelly and Shechtman, Eli and Ramamoorthi, Ravi and Chandraker, Manmohan},
  booktitle={Proceedings of the IEEE/CVF International Conference on Computer Vision},
  pages={14214--14223},
  year={2021}
}

@inproceedings{liu2024finer,
  title={Finer: Flexible spectral-bias tuning in implicit neural representation by variable-periodic activation functions},
  author={Liu, Zhen and Zhu, Hao and Zhang, Qi and Fu, Jingde and Deng, Weibing and Ma, Zhan and Guo, Yanwen and Cao, Xun},
  booktitle={Proceedings of the IEEE/CVF Conference on Computer Vision and Pattern Recognition},
  pages={2713--2722},
  year={2024}
}

@article{marszalek2025hypernetwork,
  title={A Hypernetwork-Based Approach to KAN Representation of Audio Signals},
  author={Marsza{\l}ek, Patryk and Rut, Maciej and Kawa, Piotr and Syga, Piotr},
  journal={arXiv preprint arXiv:2503.02585},
  year={2025}
}

@inproceedings{choudhury2024nerva,
  title={NeRVA: Joint Implicit Neural Representations for Videos and Audios},
  author={Choudhury, Anustup and Singh, Praneet and Su, Guan-Ming},
  booktitle={2024 IEEE International Conference on Multimedia and Expo (ICME)},
  pages={1--6},
  year={2024},
  organization={IEEE}
}

@inproceedings{kim2023regression,
  title={Regression to Classification: Waveform Encoding for Neural Field-Based Audio Signal Representation},
  author={Kim, TaeSoo and Rho, Daniel and Lee, Gahui and Park, JaeHan and Ko, Jong Hwan},
  booktitle={ICASSP 2023-2023 IEEE International Conference on Acoustics, Speech and Signal Processing (ICASSP)},
  pages={1--5},
  year={2023},
  organization={IEEE}
}

@inproceedings{kim2023generalizable,
  title={Generalizable implicit neural representations via instance pattern composers},
  author={Kim, Chiheon and Lee, Doyup and Kim, Saehoon and Cho, Minsu and Han, Wook-Shin},
  booktitle={Proceedings of the IEEE/CVF Conference on Computer Vision and Pattern Recognition},
  pages={11808--11817},
  year={2023}
}

@article{li2024asmr,
  title={Asmr: Activation-sharing multi-resolution coordinate networks for efficient inference},
  author={Li, Jason Chun Lok and Luo, Steven Tin Sui and Xu, Le and Wong, Ngai},
  journal={arXiv preprint arXiv:2405.12398},
  year={2024}
}

@inproceedings{saragadam2023wire,
title={WIRE: Wavelet Implicit Neural Representations},
author={Saragadam, Vishwanath and LeJeune, Daniel and Tan, Jasper and Balakrishnan, Guha and Veeraraghavan, Ashok and Baraniuk, Richard G},
booktitle={Conf. Computer Vision and Pattern Recognition},
year={2023}
}

@inproceedings{ramasinghe2022beyond,
  title={Beyond periodicity: Towards a unifying framework for activations in coordinate-mlps},
  author={Ramasinghe, Sameera and Lucey, Simon},
  booktitle={European Conference on Computer Vision},
  pages={142--158},
  year={2022},
  organization={Springer}
}

@article{zhu2024finerplusplus,
    title={FINER++: Building a Family of Variable-periodic Functions for Activating Implicit Neural Representation},
    author={Zhu, Hao and Liu, Zhen and Zhang, Qi and Fu, Jingde and Deng, Weibing and Ma, Zhan and Guo, Yanwen and Cao, Xun},
    url={https://arxiv.org/abs/2407.19434}, 
    year={2024},
}

@article{ronen2019convergence,
  title={The convergence rate of neural networks for learned functions of different frequencies},
  author={Ronen, Basri and Jacobs, David and Kasten, Yoni and Kritchman, Shira},
  journal={Advances in Neural Information Processing Systems},
  volume={32},
  year={2019}
}

@article{bietti2019inductive,
  title={On the inductive bias of neural tangent kernels},
  author={Bietti, Alberto and Mairal, Julien},
  journal={Advances in Neural Information Processing Systems},
  volume={32},
  year={2019}
}

@article{cao2019towards,
  title={Towards understanding the spectral bias of deep learning},
  author={Cao, Yuan and Fang, Zhiying and Wu, Yue and Zhou, Ding-Xuan and Gu, Quanquan},
  journal={arXiv preprint arXiv:1912.01198},
  year={2019}
}

@article{tang2025structured,
  title={Structured First-Layer Initialization Pre-Training Techniques to Accelerate Training Process Based on $\varepsilon$-Rank},
  author={Tang, Tao and Yang, Jiang and Zhao, Yuxiang and Zhu, Quanhui},
  journal={arXiv preprint arXiv:2507.11962},
  year={2025}
}

@article{varre2023spectral,
  title={On the spectral bias of two-layer linear networks},
  author={Varre, Aditya Vardhan and Vladarean, Maria-Luiza and Pillaud-Vivien, Loucas and Flammarion, Nicolas},
  journal={Advances in Neural Information Processing Systems},
  volume={36},
  pages={64380--64414},
  year={2023}
}

@misc{kodak_photocd,
  author = {Kodak},
  title = {Kodak PhotoCD Image Dataset},
  year = {1999},
  howpublished = {\url{https://service.tib.eu/ldmservice/dataset/kodak-photocd-image-dataset}},
  note = {Accessed: 2025-08-14}
}

@inproceedings{cimpoi14describing,
  author = {Mircea Cimpoi and Subhransu Maji and Iasonas Kokkinos and Sammy Mohamed and Andrea Vedaldi},
  title = {Describing Textures in the Wild},
  booktitle = {Proceedings of the IEEE Conference on Computer Vision and Pattern Recognition (CVPR)},
  year = {2014},
  pages = {3606--3613},
  publisher = {IEEE},
  doi = {10.1109/CVPR.2014.459},
  url = {https://www.robots.ox.ac.uk/~vgg/data/dtd/}
}

@article{chandravamsi2024high,
  title={High resolution optimized high-order schemes for discretization of non-linear straight and mixed second derivative terms},
  author={Chandravamsi, Hemanth and Frankel, Steven H},
  journal={Journal of Computational Physics},
  volume={513},
  pages={113170},
  year={2024},
  publisher={Elsevier}
}

@article{kurganov2002solution,
  title={Solution of two-dimensional Riemann problems for gas dynamics without Riemann problem solvers},
  author={Kurganov, Alexander and Tadmor, Eitan},
  journal={Numerical Methods for Partial Differential Equations: An International Journal},
  volume={18},
  number={5},
  pages={584--608},
  year={2002},
  publisher={Wiley Online Library}
}

@inproceedings{batson2019noise2self,
  title={Noise2self: Blind denoising by self-supervision},
  author={Batson, Joshua and Royer, Loic},
  booktitle={International conference on machine learning},
  pages={524--533},
  year={2019},
  organization={PMLR}
}

@article{wang2004image,
  title={Image quality assessment: from error visibility to structural similarity},
  author={Wang, Zhou and Bovik, Alan C and Sheikh, Hamid R and Simoncelli, Eero P},
  journal={IEEE transactions on image processing},
  volume={13},
  number={4},
  pages={600--612},
  year={2004},
  publisher={IEEE}
}

@inproceedings{zhang2018unreasonable,
  title={The unreasonable effectiveness of deep features as a perceptual metric},
  author={Zhang, Richard and Isola, Phillip and Efros, Alexei A and Shechtman, Eli and Wang, Oliver},
  booktitle={Proceedings of the IEEE conference on computer vision and pattern recognition},
  pages={586--595},
  year={2018}
}

@article{ding2020image,
  title={Image quality assessment: Unifying structure and texture similarity},
  author={Ding, Keyan and Ma, Kede and Wang, Shiqi and Simoncelli, Eero P},
  journal={IEEE transactions on pattern analysis and machine intelligence},
  volume={44},
  number={5},
  pages={2567--2581},
  year={2020},
  publisher={IEEE}
}

@inproceedings{novello2025tuning,
  title={Tuning the Frequencies: Robust Training for Sinusoidal Neural Networks},
  author={Novello, Tiago and Aldana, Diana and Araujo, Andre and Velho, Luiz},
  booktitle={Proceedings of the Computer Vision and Pattern Recognition Conference},
  pages={3071--3080},
  year={2025}
}

@article{yeom2024fast,
  title={Fast training of sinusoidal neural fields via scaling initialization},
  author={Yeom, Taesun and Lee, Sangyoon and Lee, Jaeho},
  journal={arXiv preprint arXiv:2410.04779},
  year={2024}
}

@article{hewa2025vi3nr,
  title={VI3NR: Variance Informed Initialization for Implicit Neural Representations},
  author={Hewa Koneputugodage, Chamin and Ben-Shabat, Yizhak and Ramasinghe, Sameera and Gould, Stephen},
  journal={arXiv e-prints},
  pages={arXiv--2504},
  year={2025}
}

@inproceedings{sopena1999neural,
  title={Neural networks with periodic and monotonic activation functions: a comparative study in classification problems},
  author={Sopena, Josep M and Romero, Enrique and Alquezar, Rene},
  booktitle={9th International Conference on Artificial Neural Networks: ICANN'99},
  pages={323--328},
  year={1999},
  organization={IET}
}

@article{parascandolo2016taming,
  title={Taming the waves: sine as activation function in deep neural networks},
  author={Parascandolo, Giambattista and Huttunen, Heikki and Virtanen, Tuomas},
  year={2016}
}

@article{jagtap2022deep,
  title={Deep Kronecker neural networks: A general framework for neural networks with adaptive activation functions},
  author={Jagtap, Ameya D and Shin, Yeonjong and Kawaguchi, Kenji and Karniadakis, George Em},
  journal={Neurocomputing},
  volume={468},
  pages={165--180},
  year={2022},
  publisher={Elsevier}
}

@inproceedings{chng2022gaussian,
  title={Gaussian activated neural radiance fields for high fidelity reconstruction and pose estimation},
  author={Chng, Shin-Fang and Ramasinghe, Sameera and Sherrah, Jamie and Lucey, Simon},
  booktitle={European Conference on Computer Vision},
  pages={264--280},
  year={2022},
  organization={Springer}
}

@inproceedings{liu2020neural,
  title={Neural Networks Fail to Learn Periodic Functions and How to Fix It},
  author={Liu, Ziyin and Hartwig, Tilman and Ueda, Masahito},
  booktitle={Advances in Neural Information Processing Systems (NeurIPS)},
  year={2020}
}

@inproceedings{saratchandran2024sampling,
  title={A Sampling Theory Perspective on Activations for Implicit Neural Representations},
  author={Saratchandran, Hemanth and Ramasinghe, Sameera and Shevchenko, Violetta and Long, Alexander and Lucey, Simon},
  booktitle={Proceedings of the 41st International Conference on Machine Learning (ICML)},
  publisher={PMLR},
  year={2024}
}

@article{morsali2025staf,
  title={STAF: Sinusoidal Trainable Activation Functions for Implicit Neural Representation},
  author={Morsali, Alireza and Vaez, MohammadJavad and Soltani, Hossein and Kazerouni, Amirhossein and Taati, Babak and Mohammad-Noori, Morteza},
  journal={arXiv preprint arXiv:2502.00869},
  year={2025}
}

@inproceedings{atzmon2020sal,
  title        = {SAL: Sign Agnostic Learning of Shapes from Raw Data},
  author       = {Atzmon, Matan and Lipman, Yaron},
  booktitle    = {Proceedings of the IEEE/CVF Conference on Computer Vision and Pattern Recognition (CVPR)},
  pages        = {2565--2574},
  year         = {2020}
}

@inproceedings{ben2022digs,
  title        = {DiGS: Divergence Guided Shape Implicit Neural Representation for Unoriented Point Clouds},
  author       = {Ben-Shabat, Yizhak and Hewa Koneputugodage, Chamin and Gould, Stephen},
  booktitle    = {Proceedings of the IEEE/CVF Conference on Computer Vision and Pattern Recognition (CVPR)},
  pages        = {19323--19332},
  year         = {2022}
}

@inproceedings{shi2024frinr,
  title     = {Improved Implicit Neural Representation with Fourier Reparameterized Training},
  author    = {Shi, Kexuan and Zhou, Xingyu and Gu, Shuhang},
  booktitle = {Proceedings of the IEEE/CVF Conference on Computer Vision and Pattern Recognition (CVPR)},
  pages     = {25985--25994},
  year      = {2024}
}

@article{shi2024iga,
  title   = {Inductive Gradient Adjustment For Spectral Bias In Implicit Neural Representations},
  author  = {Shi, Kexuan and Chen, Hai and Zhang, Leheng and Gu, Shuhang},
  journal = {arXiv preprint arXiv:2410.13271},
  year    = {2024},
  note    = {Accepted to ICML 2025}
}

@inproceedings{yuce2022structured,
  title={A structured dictionary perspective on implicit neural representations},
  author={Y{\"u}ce, Gizem and Ortiz-Jim{\'e}nez, Guillermo and Besbinar, Beril and Frossard, Pascal},
  booktitle={Proceedings of the IEEE/CVF Conference on Computer Vision and Pattern Recognition},
  pages={19228--19238},
  year={2022}
}

@article{zhu2024disorder,
  title={Disorder-invariant implicit neural representation},
  author={Zhu, Hao and Xie, Shaowen and Liu, Zhen and Liu, Fengyi and Zhang, Qi and Zhou, You and Lin, Yi and Ma, Zhan and Cao, Xun},
  journal={IEEE Transactions on Pattern Analysis and Machine Intelligence},
  volume={46},
  number={8},
  pages={5463--5478},
  year={2024},
  publisher={IEEE}
}

@inproceedings{glorot2010understanding,
  title={Understanding the difficulty of training deep feedforward neural networks},
  author={Glorot, Xavier and Bengio, Yoshua},
  booktitle={Proceedings of the thirteenth international conference on artificial intelligence and statistics},
  pages={249--256},
  year={2010},
  organization={JMLR Workshop and Conference Proceedings}
}
\pagebreak
\newpage
\newpage
\newpage

\section*{Biographies}
\begin{IEEEbiographynophoto}{Hemanth Chandravamsi} received the Ph.D. degree in mechanical engineering from the Technion -- Israel Institute of Technology, Haifa, Israel. He is currently a Postdoctoral Researcher with the Faculty of Mechanical Engineering at the Technion. His research interests include scientific and physics aware machine learning, high speed flows, and numerical algorithms.
\end{IEEEbiographynophoto}
\begin{IEEEbiographynophoto}{Dhanush V. Shenoy} received the Ph.D. degree in Aerospace Engineering from ISAE-SUPAERO, Toulouse, France, where his doctoral research focused on aeroacoustics and low–Reynolds-number flows. He is currently pursuing post-doctoral research in deep reinforcement learning and its applications to computational fluid dynamics. His research interests include machine learning for flow modeling and control, quantum computing and quantum-inspired algorithms, lattice Boltzmann methods, high-performance computing, and GPU-accelerated numerical solvers.
\end{IEEEbiographynophoto}
\begin{IEEEbiographynophoto}{Itay Zinn} received the M.Sc. degree in mechanical engineering from the Ariel University, Israel, and he is currently pursuing his Ph.D. degree in computational fluid dynamics at Technion -- Israel Institute of Technology, Israel. His research interests include machine learning, high speed flows, and computational fluid dynamics.
\end{IEEEbiographynophoto}
\begin{IEEEbiographynophoto}{Ziv Chen} received the B.Sc. degree from the Technion -- Israel Institute of Technology, Haifa, Israel, where he is currently pursuing a direct Ph.D. degree with the Faculty of Electrical and Computer Engineering. He is also affiliated with the Helen Diller Quantum Center, Technion -- Israel Institute of Technology, from which he received an excellence scholarship. His research interests include scientific and physics aware machine learning and quantum computing algorithms for physics simulations.
\end{IEEEbiographynophoto}
\begin{IEEEbiographynophoto}{Shimon Pisnoy} received the Ph.D. degree in mechanical engineering from the Technion -- Israel Institute of Technology, Haifa, Israel, where he is currently a Ph.D. student with the Faculty of Mechanical Engineering. His research interests include machine learning, high speed flows, and computational fluid dynamics.
\end{IEEEbiographynophoto}
\begin{IEEEbiographynophoto}{Steven H. Frankel} 
is a Rosenblatt Chair Professor with the Faculty of Mechanical Engineering, Technion -- Israel Institute of Technology, Haifa, Israel. He is also affiliated with the Helen Diller Quantum Center, Technion -- Israel Institute of Technology. His research interests include scientific machine learning, computational fluid dynamics, combustion, and quantum computing for physics simulations.
\end{IEEEbiographynophoto}

\vfill

\end{document}